\documentclass[11pt,letterpaper]{article}

\usepackage{amsmath,amssymb,amsthm,mathtools}
\usepackage{setspace,footmisc,indentfirst}

\usepackage{microtype}
\usepackage{booktabs}
\usepackage{array}
\usepackage{enumitem}
\usepackage[authoryear,round]{natbib}
\usepackage[hypertexnames=false]{hyperref}
\usepackage{xcolor}
\usepackage{graphicx}
\usepackage{float}
\usepackage{tikz}
\usetikzlibrary{arrows.meta,positioning}

\hypersetup{
  colorlinks=true,
  linkcolor=blue,
  citecolor=blue,
  urlcolor=blue,
  pdftitle={E-Backtesting Expected Shortfall without Value-at-Risk Forecasts under Time-series Models},
  pdfauthor={Ruodu Wang; Jingcheng Yu},
  pdfsubject={Sequential backtesting of Expected Shortfall forecasts without Value-at-Risk forecasts},
  pdfkeywords={anytime-valid inference, conditional quantile, e-process, nuisance parameter, sequential testing}
}

\newtheorem{proposition}{Proposition}

\theoremstyle{definition}
\newtheorem{remark}{Remark}
\newtheorem{example}{Example}
\theoremstyle{plain}
\newtheorem{definition}{Definition}

\newcommand{\E}{\mathbb E}
\newcommand{\Pp}{\mathbb P}
\newcommand{\VaR}{\operatorname{VaR}}
\newcommand{\ES}{\operatorname{ES}}

\newcommand{\1}{\mathbf 1}
\newcommand{\papertablestyle}{%
  \linespread{1}\normalsize\selectfont
  \setlength{\tabcolsep}{4pt}%
  \renewcommand{\arraystretch}{1.12}}
\newcommand{\tablemethod}[2]{\(#1#2\)}
\newcommand{\ProofInAppendix}[1]{%
  \par\nobreak\noindent We give the proof in
  \hyperref[#1]{Appendix~\ref*{app:proofs}}.\par}

\title{E-Backtesting Expected Shortfall without Value-at-Risk Forecasts under Time-series Models}
\author{%
Ruodu Wang\thanks{Department of Statistics and Actuarial Science,
University of Waterloo, Waterloo, Ontario, Canada.
Corresponding author: \texttt{wang@uwaterloo.ca}.}
\and
Jingcheng Yu\thanks{Department of Statistics and Actuarial Science,
University of Waterloo, Waterloo, Ontario, Canada.
Email: \texttt{j563yu@uwaterloo.ca}.}}
\date{October 12, 2026}

\begin{document}
\maketitle

\begin{abstract}
Expected Shortfall (ES) measures the average size of tail losses. Its
sequential backtests typically also require a Value-at-Risk (VaR) forecast.
We develop two e-backtesting methods using only observed losses and ES
forecasts: one for iid losses only and the other for time series under an
AR-GARCH conditional VaR structure.
Under the stated assumptions, both control false rejections at any stopping
time. For the time-series method, validity requires a fixed
parameter--quantile pair representing the conditional VaR path and a positive
report margin. The time-series statistic
takes an infimum over possible parameter and VaR values. We choose the betting
process \((\lambda_t)\) by maximizing an empirical growth rate averaged
across candidate parameters. We compute the statistic by continuous VaR
profiling on a dense parameter grid. For grid approximations, we prove
convergence to the continuous infimum, together with convergence of
the probability of false rejection, as the grid is refined over a fixed
monitoring period. Lower and upper bounds for the VaR infimum on the
parameter grid control the accuracy of the numerical profiling.
Simulations, including ES
reports that change from over-reporting to under-reporting, and an application
to daily NASDAQ losses provide numerical comparisons with VaR--ES e-processes
and fixed-sample ES backtests, assessing detection frequencies and the timing
of rejections.


\end{abstract}

\section{Introduction}

Expected Shortfall (ES) is the principal measure of tail risk in the internal-model
component of the Basel market-risk framework \citep{BASEL19}.
It measures the average size of tail losses.  Backtesting
ES forecasts therefore requires assessing whether the reported amounts
adequately reflect the size of potential tail losses.  When losses arrive
sequentially, a regulator may want to act as soon as evidence of
under-reporting emerges.  This calls for a procedure that controls false
rejections at any stopping time.

VaR plays a central role in many existing methods for evaluating ES forecasts.
ES alone is not elicitable \citep{gneiting-2011}, whereas VaR and ES are
jointly elicitable \citep{fissler-ziegel-2016}:  they admit scoring
functions whose expected value is uniquely minimized by the true forecast
pair.  For sequential ES backtesting, the ES e-statistic
of \citet{wang-wang-ziegel-2026} measures the loss above a reported VaR
relative to the reported ES--VaR gap.
Its null requires the reported VaR to equal the true conditional VaR, which is a non-trivial   assumption.
This leads to  a practically important question: how can ES forecasts be backtested sequentially when VaR
forecasts are unavailable?

To get rid of the reliance on the VaR forecasts, some price needs to be paid. 
This boils down to, essentially, a way to calibrate some information on the VaR form data, which reflects the size of the tail losses. 
For this, we consider
two settings: iid losses as a benchmark, and an AR-GARCH setting for
practical implementation, because AR-GARCH is one of the most popular models for financial data and it is simple enough to interpret; see e.g., \citet[Chapter 4]{MFE15}. 
For iid losses, calibration data give a confidence set for VaR.
For time series, we assume a fixed but unknown conditional VaR for the standardized residuals (innovations). This assumption is weaker than the assumption of iid innovations commonly used for time-series models. 
The key idea is very simple: we tackle the unreported VaR by taking the infimum of cumulative ES betting products over feasible values of the AR-GARCH parameters and the residual VaR.

There are 
existing ES backtests, which differ in their assumptions on the loss distribution
and temporal dependence. 
\citet{mcneil-frey-2000} leave the innovation distribution unspecified in
their bootstrap backtest, but rely on a location--scale representation with
iid innovations and correctly specified conditional VaR and volatility
forecasts to test the ES forecasts.  \citet{du-escanciano-2017} specify a parametric
conditional distribution and assume stationarity and ergodicity; their
statistics use this distribution to measure tail violations. 
Without specifying a full forecasting model, \citet{moldenhauer-pitera-2019}
use only ES forecasts and realized losses. They study the ES forecast minus
the realized loss, assume iid behaviour in the relevant tail of these
differences, and assess their rejection thresholds under selected
distributions. The strict and intercept ES regression (ESR) backtests of
\citet{bayer-dimitriadis-2022} jointly estimate quantile and ES regressions
using only ES forecasts and realized losses. Both ESR backtests allow dependence that
decays with time separation (strong mixing), but their asymptotic theory
requires finite higher moments and smooth conditional densities.
Misspecification of the quantile regression can also shift the null values
of the ES regression coefficients.
 
 The ES-only backtests
discussed above evaluate a fixed sample and do not provide an anytime-valid e-process for monitoring
conditional ES forecasts.
\citet[Proposition~3]{wang-wang-ziegel-2026} show that, over all distributions
with finite first moment, no nonnegative statistic of a single loss and its
ES forecast can have expectation at most one under every null distribution
and greater than one under every under-reporting alternative.
Simply counting losses above the ES forecast also misses
under-reporting caused by larger tail losses with unchanged exceedance
probability.

Several related studies address sequential inference and ES backtesting.
\citet{casgrain-larsson-ziegel-2024} construct test supermartingales for
elicitable and identifiable functionals, including the pair of VaR and ES. Their Appendix A has sketched 
the idea of taking infimum over unrestricted parameters, a principle that our methods also build on. 
\citet{howard-ramdas-mcauliffe-sekhon-2021} develop confidence sequences
with simultaneous coverage over time, and \citet{howard-ramdas-2022}
provide distribution-free quantile confidence sequences for iid
observations. \citet{agrawal-koolen-juneja-2021} obtain anytime-valid CVaR
confidence bounds under a known higher-moment bound in a best-arm
identification problem. Among other ES backtests, the auxiliary ESR test
of \citet{bayer-dimitriadis-2022} requires both VaR and ES forecasts.
\citet{hoga-demetrescu-2023} study sequential monitoring of VaR and ES
forecasts, while \citet{hue-hurlin-lu-2024} use probability integral
transforms to assess the duration and severity of VaR violations when
backtesting ES.
Table~\ref{tab:related-literature} compares the forecasts used, model
assumptions, and validity of the closest e-backtesting methods and ES
backtests. The main null of \citet{wang-wang-ziegel-2026} includes a correct
conditional VaR forecast, without a parametric loss model or restrictions on
temporal dependence. Our time-series construction uses an AR-GARCH structure
for the unreported conditional VaR path. This structure must represent the
true conditional VaR throughout monitoring for some fixed feasible pair
\((\theta_0,z_0)\); it does not specify the innovation distribution or require
iid residuals.

\begin{table}[!htbp]
\centering
\caption{E-backtesting and related ES backtests.}
\label{tab:related-literature}
\papertablestyle
\begin{tabular}{@{}>{\raggedright\arraybackslash}p{0.13\linewidth}
>{\raggedright\arraybackslash}p{0.21\linewidth}
>{\raggedright\arraybackslash}p{0.39\linewidth}
>{\raggedright\arraybackslash}p{\dimexpr0.27\linewidth-24pt\relax}@{}}
\toprule
Method & Forecasts used & Model assumptions & Validity \\
\midrule
WWZ26
& ES and VaR
& No parametric or dependence restrictions; correct conditional VaR
& Anytime valid \\
CLZ24
& Fixed values of VaR and ES
& Fixed conditional VaR and ES; bounded identification function or sub-\(\psi\) bound
& Anytime valid \\
\midrule
BD22
& ES only (strict and intercept ESR)
& Strong mixing and VaR--ES regression regularity
& Asymptotic, fixed sample \\
\midrule
This paper: iid
& ES only, constant over time
& iid losses; separate calibration sample
& Anytime valid \\
This paper: time series
& ES only, allowed to vary over time
& AR-GARCH conditional VaR structure; report margin
& Anytime valid \\
\bottomrule
\end{tabular}

\vspace{2pt}
\begin{minipage}{\linewidth}
\raggedright
\emph{Notes:} WWZ26, CLZ24, and BD22 denote \citet{wang-wang-ziegel-2026},
\citet{casgrain-larsson-ziegel-2024}, and \citet{bayer-dimitriadis-2022}.
\end{minipage}
\end{table}

The tester observes losses \(L_t\) and ES forecasts \(r_t\) at level
\(p\in(0,1)\).  Each report is measurable with respect to the information
\(\mathcal F_{t-1}\) available before observing \(L_t\).  We construct
e-processes for the null
\(\ES_p(L_t\mid\mathcal F_{t-1})\le r_t\) throughout the monitoring period.
For a threshold \(c>1\), a \emph{detection} means that the monitored statistic
reaches or exceeds \(c\) at least once during monitoring. For an e-process,
choosing \(c=1/\alpha\) gives a level-\(\alpha\) test at any stopping time.

We develop two ES-only e-backtests: one for iid losses with a constant ES
forecast, and one for time series under an AR-GARCH conditional VaR structure.
Both yield e-processes under the stated null hypotheses and model conditions.
We separate the costs of parameter and VaR profiling using the same loss path
and betting process, and bound the loss in log wealth from VaR profiling.
We also study the loss of maximal expected log growth and give conditions for
asymptotic power one at the true parameter, with VaR still profiled.
We establish convergence of grid approximations over fixed finite monitoring
periods. Simulations and a NASDAQ application compare the implemented
procedures with VaR--ES e-processes and fixed-sample ES backtests.

The rest of the paper is organized as follows.
Section~\ref{sec:iid-esonly} constructs an ES-only e-process for iid losses.
Section~\ref{sec:argarch-profile} extends the construction to time-series
losses and quantifies the cost of an unreported VaR.
Section~\ref{sec:betting-selection} explains how we choose the betting
process \((\lambda_t)\), and Section~\ref{sec:profile-computation} describes
the dense-grid approximation and proves its convergence.
Sections~\ref{sec:sim} and~\ref{sec:financial-data} present simulations and
the NASDAQ application, respectively. Section~\ref{sec:conclusion} concludes.
Appendix~\ref{app:grid-sensitivity} examines sensitivity to grid refinement,
Appendix~\ref{app:proofs} contains the proofs, and
Appendix~\ref{app:gamed} examines ES forecasts that change from
over-reporting to under-reporting.

\section{From the ES e-statistic to an iid ES-only e-process}
\label{sec:iid-esonly}

We first recall the standard ES e-statistic, which uses both the reported ES
and VaR values.  We then turn it into an ES-only e-process for iid losses by using observations from a calibration period
to restrict the unreported VaR and profiling over the resulting confidence
set.

\subsection{The ES e-statistic with an observed VaR}
\label{sec:es-estatistic}

Fix a probability level \(p\in(0,1)\). For a distribution \(F\), let
\(\VaR_p(F)=\inf\{x:F(x)\ge p\}\) and let \(\ES_p(F)\) denote
the Expected Shortfall at level \(p\).
Writing \(z_0=\VaR_p(F)\), we have \citep{rockafellar-uryasev-2002}
\begin{equation}
\label{eq:ru-lower-quantile}
\ES_p(F)
=
z_0+\frac{1}{1-p}\int(x-z_0)_+\,\mathrm dF(x).
\end{equation}
We allow the value \(\ES_p(F)=\infty\); it is finite if and only if \(\int x_+\,\mathrm dF(x)<\infty\).
VaR and ES are closely related in many ways; for 
the conversion between them in common distribution models, see \cite{LW23}.

For a hypothesis \(\mathcal P\), viewed as a set of probability measures, a
nonnegative random variable \(E\) is an e-variable for \(\mathcal P\) if
\(\sup_{P\in\mathcal P}\E_P[E]\le1\).  Its observed value is an e-value
\citep{vovk-wang-2021}.  An adapted nonnegative process \((M_t)\) is an
e-process for \(\mathcal P\) if
\(\sup_{P\in\mathcal P}\E_P[M_\tau]\le1\) for every stopping time \(\tau\) in
the monitoring period; the underlying filtration is assumed to be that of the data process.  By Markov's inequality at the first time the process
reaches \(1/\alpha\),
\(\sup_{P\in\mathcal P}\Pp_P(\max_t M_t\ge1/\alpha)\le\alpha\) on
the same period. We use the ES e-statistic of \citet{wang-wang-ziegel-2026} defined by
\[
e_p^{\mathrm{ES}}(x,r,z)
=
\frac{(x-z)_+}{(1-p)(r-z)},
\qquad z\le r.
\]
The numerator is the realized excess loss above \(z\), and the denominator
is the expected excess loss implied by the reported ES--VaR pair \((r,z)\).
At the boundary \(z=r\), following the convention of
\citet{wang-wang-ziegel-2026}, we set \(0/0=1\) and
\(a/0=\infty\) for \(a>0\).
If
\(H_0:\ES_p(F)\le r\) holds, then \(z_0\le r\).  If \(z_0<r\),
\eqref{eq:ru-lower-quantile} gives \(\int(x-z_0)_+\,\mathrm dF(x)\le (1-p)(r-z_0)\).  If \(z_0=r\), the inequalities
\(z_0\le\ES_p(F)\le r\) imply \(\ES_p(F)=r\), and
\eqref{eq:ru-lower-quantile} implies \(F(r)=1\), so the boundary e-statistic equals one for \(F\)-almost every \(x\).  Thus, in both cases, \(\int e_p^{\mathrm{ES}}(x,r,z_0)\,\mathrm dF(x)\le1\), and for any \(u\in[0,1-p]\), \(\int\{1-u+u\,e_p^{\mathrm{ES}}(x,r,z_0)\}\,\mathrm dF(x)\le1\).
For iid observations, products of these factors
are e-values. As each new loss is observed, we update the cumulative statistic
as follows:
\[
M_t(z;\lambda)
=\prod_{s=1}^t\{1-\lambda_s+\lambda_s e_p^{\mathrm{ES}}(L_s,r_s,z_s)\},
\]
where \(\lambda=(\lambda_t)\) is the predictable betting process and
\(z=(z_s)\) is the sequence of reported VaR forecasts, which this subsection
takes as observed.
If each \(z_s\) is the true conditional VaR and each \(r_s\) satisfies the null, the
conditional mean of every factor is at most one, and \(M_t(z;\lambda)\) is a
nonnegative supermartingale starting at one, and hence an e-process.  Note that at a
\(z\) that is not the true VaR, no such statement holds.
Throughout our ES-only constructions, we use
\(\lambda_t\in[0,1-p]\). This keeps every betting factor at least
\(p>0\), so the capital stays positive and its logarithm, which we use to
measure growth, is well defined.

Simply inserting an estimated VaR into the ES e-statistic does not in
general give an e-process for the ES null alone. To see this, suppose
\(F\) is continuous and strictly increasing and \(r=\ES_p(F)<\infty\).
Let \(\widehat z<r\) be a VaR estimate from an independent calibration
sample and let \(L\sim F\) be the first monitored loss. By the ES
minimization formula of \citet{rockafellar-uryasev-2002},
\[
\E_F[e_p^{\mathrm{ES}}(L,r,\widehat z)\mid\widehat z]
=1+\frac{\widehat z+\dfrac{1}{1-p}\int(x-\widehat z)_+\,\mathrm dF(x)-r}
{r-\widehat z}
\ge1,
\]
with strict inequality whenever \(\widehat z\ne z_0\). If this occurs
with positive probability, the first capital factor has expectation
greater than one for any fixed \(\lambda_1\in(0,1-p]\). Thus the product
is not an e-process, even though the ES report is correct and the VaR
estimate uses only past data.

\begin{remark}
The ES null implies
\(\Pp(L_t>r_t\mid\mathcal F_{t-1})\le1-p\), because ES is at least
as large as VaR at the same level. Define the ES-only exceedance process by
\begin{equation}
\label{eq:es-exceedance-process}
\prod_{s=1}^t\left\{1-\lambda_s+\lambda_s\frac{\1_{\{L_s>r_s\}}}{1-p}\right\},
\qquad t\ge0.
\end{equation}
Under the ES null, this process is a nonnegative supermartingale for every
predictable \(\lambda_t\in[0,1-p]\). The current loss enters each factor
only through the indicator \(\1_{\{L_t>r_t\}}\). Larger losses above the
forecast can therefore increase ES without increasing the exceedance
probability. Under-reporting alternatives satisfying
\(\Pp(L_t>r_t\mid\mathcal F_{t-1})\le1-p\) throughout monitoring also make
this process a nonnegative supermartingale. Its probability of ever reaching
\(c>1\) is then at most \(1/c\), limiting its power against these alternatives.
\end{remark}

\subsection{An ES-only e-process for iid losses}
\label{sec:confidence-profile}

We now consider the iid case, which is not practical for financial applications but is helpful for us to build intuition for the general methodology. 

Under iid losses, the unknown VaR is a single constant.  We use the first
\(n_c\) observations only to form a one-sided confidence set for this value.
Once the confidence set is fixed, the monitoring process takes the infimum over the VaR values in that
set.

We use \(t=1,\ldots,n_c\)
for the calibration period and \(t=n_c+1,\ldots,T\) for the monitoring period,
which has \(T-n_c\) observations.
Let \(L_{(1)}\le \cdots \le L_{(n_c)}\) be the order statistics of
\(L_1,\ldots,L_{n_c}\).  For a calibration error level \(\delta\), choose \(j_\delta\) by the formula
below. This index determines the VaR lower confidence bound
\(z_\delta=L_{(j_\delta)}\), with \(L_{(0)}=-\infty\).
\[
j_\delta
=
\max\left\{j\in\{0,\ldots,n_c\}:\Pp\{\operatorname{Bin}(n_c,p)<j\}\le\delta\right\}.
\]
This index is approximately \(n_c p\), shifted downward until the binomial
tail probability is at most \(\delta\).  For every loss distribution \(F\),
\(\Pp_F(z_\delta>z_0) \le \Pp\{\operatorname{Bin}(n_c,p)<j_\delta\} \le \delta\).

Let
\(\mathcal F_t=\sigma(L_1,\ldots,L_t)\), and for \(t>n_c\) let
\(\lambda_t\in[0,1-p]\) be any \(\mathcal F_{t-1}\)-measurable choice.  Fix \(\Gamma\ge1\).  For \(t\ge n_c\) and \(z\le r\), let
\begin{equation}
M_t(z;\lambda)
=
\prod_{s=n_c+1}^t
\left\{
1-\lambda_s+\lambda_s e_p^{\mathrm{ES}}(L_s,r,z)
\right\},
\end{equation}
where a factor with \(\lambda_s=0\) is defined to be one, including when
the e-statistic is infinite.  In the infimum below, the boundary value \(z=r\) is evaluated explicitly, in addition to
the interior values in \([z_\delta,r)\).
On the event \(\{z_\delta\le r\}\), define
\begin{equation}
\label{eq:calibrated-process}
M_t^{\mathrm{cal}}(\lambda)
=
\frac{\min\{\inf_{z_\delta\le z\le r}M_t(z;\lambda),\,\Gamma\}}{1+\Gamma\delta},
\qquad t\ge n_c.
\end{equation}
On the complementary event \(\{z_\delta>r\}\), the lower confidence
bound for VaR from observations in the calibration period exceeds the reported ES.
Since \(\VaR_p(F)\le\ES_p(F)\), this is already evidence against the
ES null.  We therefore set \(M_t^{\mathrm{cal}}(\lambda)=\Gamma/(1+\Gamma\delta)\), the largest
value allowed by the capped construction, for all \(n_c\le t\le T\).  Under the
null in Proposition~\ref{thm:iid-construction} below, \(z_0\le r\), so this
event is contained in the coverage failure \(\{z_\delta>z_0\}\) and has
probability at most \(\delta\).  It did not occur in the reported
simulations or in the empirical analysis.

We call \(M_t^{\mathrm{cal}}(\lambda)\) the confidence-bound e-process.  In the simulations, the betting process \((\lambda_t)\) is chosen using
observations from the calibration period and kept constant during monitoring. The validity result
allows any predictable betting process \((\lambda_t)\).
Truncation at \(\Gamma\) controls the contribution from coverage failure
\(\{z_\delta>z_0\}\).  On the complement of this event,
\(z_0\in[z_\delta,r]\), so the infimum in \eqref{eq:calibrated-process} is at most \(M_t(z_0;\lambda)\).  On the coverage-failure event, this domination need not hold, but the
numerator in \eqref{eq:calibrated-process} is at most \(\Gamma\) in either
case of the construction.  The normalizing factor
\(1+\Gamma\delta\) accounts for this term in the expectation bound.

\begin{remark}
For \(M_t^{\mathrm{cal}}(\lambda)\) we set \(\delta=0.001\) and \(\Gamma=40\).
We divide the capped statistic by \(1+\Gamma\delta=1.04\), giving a maximum
reported e-value of \(\Gamma/(1+\Gamma\delta)\approx38.46\), above the
largest threshold \(c=10\). Decreasing \(\delta\) widens the
confidence interval, and increasing \(\Gamma\) increases the divisor
\(1+\Gamma\delta\).
\end{remark}

Suppose that \(L_1,L_2,\ldots\) are iid with distribution \(F\) and
\(\ES_p(F)\le r\). Fix a finite terminal time \(T\ge n_c\).

\begin{proposition}
\label{thm:iid-construction}
The process \(M_t^{\mathrm{cal}}(\lambda)\) in \eqref{eq:calibrated-process} is an e-process under \(F\);
that is, for every
stopping time \(n_c\le\tau\le T\),
\(\E_F[M_\tau^{\mathrm{cal}}(\lambda)]\le1\).
Hence, for every \(c>1\),
\[
\Pp_F
\left(
\sup_{n_c\le t\le T}M_t^{\mathrm{cal}}(\lambda)\ge c
\right)
\le
\frac1c.
\]
\end{proposition}

\ProofInAppendix{proof:thm-iid-construction}

The proposition does not require \(\VaR_p(F)<\ES_p(F)\).  The inequality
\(\VaR_p(F)\le\ES_p(F)\le r\) ensures that the true VaR belongs to \([z_\delta,r]\) on every successful coverage event.  In the boundary case
\(\VaR_p(F)=\ES_p(F)=r\), the true-VaR e-statistic is identically one under
the null.
If the process is defined on an infinite sequence, the
corresponding statement follows from the same result applied to
\(\tau\wedge T\) and then letting \(T\to\infty\).

In the simulations, we choose \(\lambda\) using only the calibration
observations and keep \(\lambda_t=\lambda\) throughout monitoring, so the
betting process is predictable. For \(j_\delta\ge1\) and \(z_\delta<r\), let
\(X_s=e_p^{\mathrm{ES}}(L_s,r,z_\delta)\), \(s=1,\ldots,n_c\). Maximizing
the second-order Taylor approximation of the empirical log growth around
\(\lambda=0\) over \([0,1-p]\) gives
\[
\lambda
=
\min\left\{1-p,\max\left\{0,
\frac{\sum_{s=1}^{n_c}(X_s-1)}
{\sum_{s=1}^{n_c}(X_s-1)^2}
\right\}\right\}.
\]
The denominator is positive because \(X_s=0\) for at least \(j_\delta\)
calibration observations. We set \(\lambda=0\) if \(j_\delta=0\) or
\(z_\delta=r\); when \(z_\delta>r\), we use the capped value defined above.
For the time-series construction, we use the
betting process \((\lambda_t)\) specified in
Section~\ref{sec:betting-selection}.

This iid construction uses the calibration error \(\delta\) of the confidence set,
whereas the time-series construction below uses a different argument, with no calibration
error for \(\Theta\): the validity proof compares the infimum with the
product obtained from one null value.
The e-power comparison in Proposition~\ref{prop:profiling-epower} below
includes the infimum over a fixed VaR range for iid losses as a special
case of the time-series construction.

\section{Time-series losses}
\label{sec:argarch-profile}

For time-series losses, we compute standardized residuals for each
candidate parameter vector \(\theta\), using the location and scale
obtained from past losses. We apply the ES e-statistic to these residuals
and the ES forecasts expressed on the same residual scale.
The losses \(L_1,\ldots,L_{n_c}\) form the calibration period and provide the initial
history of these residuals and transformed ES forecasts used to choose \(\lambda_t\).
Monitoring starts with capital one. At each date \(t\), we first receive
the ES report \(r_t\), choose \(\lambda_t\) using this report and past losses,
and update the feasible parameter--VaR pairs. After observing \(L_t\), we
update their cumulative products and compare the minimum over the feasible
grid pairs with the detection threshold.
Throughout the remaining sections, \(z\) denotes a residual VaR;
the corresponding loss-scale VaR at parameter \(\theta\) is
\(\mu_t(\theta)+\sigma_t(\theta)z\).  At a fixed pair \((\theta_0,z_0)\) satisfying
\eqref{eq:common-conditional-var} below, this loss-scale VaR is the
conditional VaR.

\subsection{Residual standardization and the null hypothesis}
\label{sec:residual-quantile-null}

Let \((\mathcal F_t)\) be the filtration of information available at time
\(t\), and let \(\Theta\) be the set of candidate parameter values
\(\theta=(\phi,\alpha,\beta,v)\), over which the e-process below takes its
infimum.  For every parameter value \(\theta\in\Theta\), we use the formulas below to
compute \(\mu_t(\theta)\) and \(\sigma_t(\theta)\) from past losses. Both
quantities are known before \(L_t\) is observed.  They are the conditional mean
and standard deviation if the innovations have conditional mean zero
and variance one; validity requires no such moment assumptions.

For \(\theta\in\Theta\), define the standardized residual and standardized ES
forecast by
\[
Z_t(\theta)=\frac{L_t-\mu_t(\theta)}{\sigma_t(\theta)},
\qquad
\tilde r_t(\theta)=\frac{r_t-\mu_t(\theta)}{\sigma_t(\theta)}.
\]
Before monitoring, we fix a report margin \(\varepsilon>0\) in these
standardized units and constants
\(\phi_{\max},\alpha_{\max},\beta_{\max}\), \(\eta_\Theta>0\), and
\(V\ge1\), and set
\[
\Theta
=
\{(\phi,\alpha,\beta,v):
|\phi|\le\phi_{\max},\ 
0\le\alpha\le\alpha_{\max},\ 
0\le\beta\le\beta_{\max},\ 
\alpha+\beta\le1-\eta_\Theta,\quad
0\le v\le V\}.
\]
The infimum ranges over the initial value \(v\) of the squared-scale recursion,
as well as the coefficients \(\phi,\alpha,\beta\). We normalize the intercept
to \(\omega=1-\alpha-\beta\), so \(\omega\ge\eta_\Theta\) on this compact set.
If
\(\omega\) were unrestricted, multiplying
\((\omega,\alpha,v)\) by a common positive constant \(a\)
multiplies every squared scale by \(a\).  Scaling \(z\) and the
report margin \(\varepsilon\) by \(a^{-1/2}\) then preserves the factors and
\(M_t(\lambda)\) over the correspondingly transformed \(\Theta\).

Let \(N_0\ge0\) be the number of additional historical losses used to run
the location and scale recursions before calibration begins. This history
reduces the influence of the initial variance \(v\) on the scales used to
form the calibration residuals and standardized ES forecasts, which enter
the choice of \(\lambda_t\). These \(N_0\) observations have indices
\(t=1-N_0,\ldots,0\); the following
\(n_c\) observations, \(t=1,\ldots,n_c\), form the calibration period,
and monitoring starts at \(t=n_c+1\). If \(N_0=0\), the recursions start
at the first calibration observation.

We also use one preceding loss \(L_{-N_0}\), not counted in \(N_0\), to
compute the first location term \(\mu_{1-N_0}(\theta)\). Set
\(u_{-N_0}(\theta)=0\) and \(\sigma_{-N_0}^2(\theta)=v\), and for
\(t\ge1-N_0\) define
\[
\mu_t(\theta)=\phi L_{t-1},\qquad
u_t(\theta)=L_t-\mu_t(\theta),\qquad
\sigma_t^2(\theta)
=
\omega+\alpha u_{t-1}(\theta)^2+\beta\sigma_{t-1}^2(\theta).
\]
Thus \(\sigma_{1-N_0}^2(\theta)=\omega+\beta v\). The recursions continue
through calibration and monitoring without a restart. The restriction
\(v\in[0,V]\) applies only to the initial squared scale, before the first
observation of the reference period comprising the historical and
calibration observations. Subsequent squared scales may exceed \(V\).
Coverage of an AR-GARCH
generating scale requires an initial value \(v_0\in[0,V]\) satisfying
\(\sigma_{1-N_0}^2(\theta_0)=\omega_0+\beta_0v_0\); the formulas above then give equality at subsequent dates as well.  Holding the observations fixed, the range
\([0,V]\) contributes at most \(\beta^{N_0+t}V\) to the squared scale at
time \(t\). Each additional historical observation reduces this bound by
a factor of \(\beta\). The effect can decay slowly when \(\beta\) is
close to one, so we retain the full range \(v\in[0,V]\) in the calculation.

\begin{remark}
We use \(\varepsilon=0.01\) throughout. For the reference calculation, the
parameter set is
\begin{equation}
\label{eq:simulation-parameter-set}
\begin{aligned}
\Theta=\{(\phi,\alpha,\beta,v):\;& |\phi|\le0.90,\quad
0\le\alpha\le0.80,\\
&0\le\beta,\quad \alpha+\beta\le0.9995,\quad 0\le v\le50\},
\end{aligned}
\end{equation}
so \(V=50\) and \(\omega=1-\alpha-\beta\ge0.0005\). This keeps \(\Theta\)
away from the boundary \(\alpha+\beta=1\), where the normalized intercept
vanishes. On this set, \(\mu_t(\theta)\) and \(\sigma_t(\theta)\) are
continuous in \(\theta\), with \(\sigma_t(\theta)>0\).
We use \(N_0=0\) in the simulations and \(N_0=700\) in the NASDAQ analysis.
The fixed parameter grid used to choose \((\lambda_t)\) is specified in
Section~\ref{sec:numerical-implementation}.
\end{remark}

For an unknown pair \((\theta_0,z_0)\in\Theta\times\mathbb R\), we assume
that the formulas above give well-defined values of \(\mu_t(\theta_0)\) and
\(\sigma_t(\theta_0)>0\) and, almost surely,
\begin{equation}
\label{eq:common-conditional-var}
\VaR_p\!\left(Z_t(\theta_0)\mid\mathcal F_{t-1}\right)=z_0,
\qquad n_c<t\le T.
\end{equation}
We call \(z_0\) the true residual VaR.  The residual distribution may change
with time and history while its conditional \(p\)-quantile remains fixed.  Condition \eqref{eq:common-conditional-var} is imposed only
during monitoring; earlier observations are used to begin computing \(\mu_t(\theta)\) and
\(\sigma_t(\theta)\) and to help determine the betting process
\((\lambda_t)\).

The conditional residual ES \(\ES_p(Z_t(\theta_0)\mid\mathcal F_{t-1})\) may
vary over time. Under the null, the finite report bounds it, giving a finite
conditional positive-tail expectation by \eqref{eq:ru-lower-quantile};
under alternatives, it may be infinite.
Conditional on \(\mathcal F_{t-1}\), the location \(\mu_t(\theta_0)\)
and positive scale \(\sigma_t(\theta_0)\) are fixed. The conditional VaR
and ES of \(L_t\) are therefore obtained by multiplying the residual VaR
\(z_0\) and the residual ES by \(\sigma_t(\theta_0)\) and adding
\(\mu_t(\theta_0)\):
\[
\begin{aligned}
\VaR_p(L_t\mid\mathcal F_{t-1})
&=\mu_t(\theta_0)+\sigma_t(\theta_0)z_0,\\
\ES_p(L_t\mid\mathcal F_{t-1})
&=\mu_t(\theta_0)+\sigma_t(\theta_0)\ES_p(Z_t(\theta_0)\mid\mathcal F_{t-1}).
\end{aligned}
\]
An AR-GARCH model whose iid innovations have distribution \(F\) is a special case,
with \(z_0=\VaR_p(F)\) and \(\ES_p(Z_t(\theta_0)\mid\mathcal F_{t-1})=\ES_p(F)\).

For each \(s>n_c\) and \(z<\tilde r_s(\theta)\), write the ES e-statistic
from Section~\ref{sec:es-estatistic} in residual units as
\begin{equation}
X_s(\theta,z)
=\frac{(Z_s(\theta)-z)_+}{(1-p)\{\tilde r_s(\theta)-z\}}.
\end{equation}
The standardized ES forecast \(\tilde r_s(\theta)\) and the value \(\lambda_s\) of
the betting process \((\lambda_t)\) are \(\mathcal F_{s-1}\)-measurable.  Each factor uses its own date's
report and remains fixed after \(L_s\) is observed.  For a fixed data path,
\(X_s\) is continuous in \((\theta,z)\) on its domain.

\begin{definition}
For \(n_c<t\le T\) and \(\theta\in\Theta\), define
\(\bar z_t(\theta) =\min_{n_c<s\le t}\tilde r_s(\theta)-\varepsilon\).
A pair \((\theta,z)\in\Theta\times\mathbb R\) is
\emph{feasible at time \(t\)} if \(z\le\bar z_t(\theta)\), or equivalently
if \(\tilde r_s(\theta)-z\ge\varepsilon\) for every \(n_c<s\le t\).
\end{definition}

The minimum restricts only the feasible residual VaR; it does not replace the
report in any factor.
Our target is an ES forecast that understates the true conditional ES at
one or more times in the monitoring period.  Under
\eqref{eq:common-conditional-var}, the finite-sample error
guarantee applies to the following null: there is a pair \((\theta_0,z_0)\)
satisfying, almost surely,
\[
H_0:\quad
\theta_0\in\Theta,\qquad
\tilde r_t(\theta_0)\ge \ES_p(Z_t(\theta_0)\mid\mathcal F_{t-1}),\qquad
z_0\le \tilde r_t(\theta_0)-\varepsilon,
\qquad n_c<t\le T.
\]

Under the null, \(\theta_0\) belongs to \(\Theta\), each ES forecast is at
least the true conditional ES, and each standardized ES forecast lies at least
\(\varepsilon\) above the true residual VaR.  The ES condition is equivalent to
\(r_t\ge\mu_t(\theta_0)+\sigma_t(\theta_0)\ES_p(Z_t(\theta_0)\mid\mathcal F_{t-1})=\ES_p(L_t\mid\mathcal F_{t-1})\).
The report margin condition is part of the statistic and its null, ensuring
\(z_0\le\bar z_t(\theta_0)\).  If the report is not below the conditional
ES, a sufficient condition follows from the conditional ES identity:
\[
\tilde r_s(\theta_0)-z_0
\ \ge\
\ES_p(Z_s(\theta_0)\mid\mathcal F_{s-1})-z_0
\ =\
\frac{\E[(Z_s(\theta_0)-z_0)_+\mid\mathcal F_{s-1}]}{1-p}.
\]
Thus \(\ES_p(Z_s(\theta_0)\mid\mathcal F_{s-1})-z_0\ge\varepsilon\) throughout the monitoring period suffices.
On this class, failure of the report margin condition implies under-reporting,
since \(\tilde r_s(\theta_0)<z_0+\varepsilon\le\ES_p(Z_s(\theta_0)\mid\mathcal F_{s-1})\).
Without this condition, even an exact ES forecast may fail the
report margin condition.  The e-process \(M_t(\lambda)\) defined below remains defined in either case
because its feasible residual VaR set contains all sufficiently negative
values.

For validity, \(\Theta\) may be any fixed nonempty set of parameter values
\(\theta=(\phi,\alpha,\beta,v)\) for which the formulas above define
\(\mu_t(\theta)\) and \(\sigma_t(\theta)\).  The ES forecast
sequence \(r_t\) and the betting process \((\lambda_t)\) are predictable,
with \(\lambda_t\in[0,1-p]\). For each fixed \(\theta\), the paths
\(\mu_t(\theta)\) and \(\sigma_t(\theta)\) are predictable; for each fixed
data path, they are continuous in \(\theta\), with \(\sigma_t(\theta)>0\).

\subsection{The e-process \texorpdfstring{\(M_t(\lambda)\)}{M\_t(lambda)}}

Let \(\lambda=(\lambda_t)\) be a predictable betting process.  Set
\(M_{n_c}(\lambda)=1\) and, for \(n_c<t\le T\), define
\begin{equation}
\label{eq:full-parameter-profile}
M_t(\lambda)
=
\inf_{\substack{\theta\in\Theta\\z\le\bar z_t(\theta)}}
\prod_{s=n_c+1}^t
\{1-\lambda_s+\lambda_s X_s(\theta,z)\}.
\end{equation}
The same pair \((\theta,z)\) is used in every factor; the product and the
constraint on \(z\) use the same observations \(n_c<s\le t\).
For a feasible pair, \(M_t(\theta,z;\lambda)\) denotes the product
inside the infimum, with \(M_{n_c}(\theta,z;\lambda)=1\).

\begin{proposition}
\label{thm:full-parameter-profile-validity}
If a fixed pair \((\theta_0,z_0)\) satisfies
\eqref{eq:common-conditional-var} and \(H_0\) holds,
then
\((M_t(\lambda))_{n_c\le t\le T}\) is an e-process: \(\E[M_\tau(\lambda)]\le1\)
for every
stopping time \(n_c\le\tau\le T\).
Hence, for every \(c>1\),
\[
\Pp\left(\sup_{n_c<t\le T}M_t(\lambda)\ge c\right)\le \frac1c.
\]
\end{proposition}

\ProofInAppendix{proof:thm-full-parameter-profile-validity}

At a pair \((\theta_0,z_0)\) satisfying \(H_0\), predictability of the report
and the conditional ES identity give
\[
\E[X_s(\theta_0,z_0)\mid\mathcal F_{s-1}]
=\frac{\ES_p(Z_s(\theta_0)\mid\mathcal F_{s-1})-z_0}{\tilde r_s(\theta_0)-z_0}\le1.
\]
Thus \(M_t(\theta_0,z_0;\lambda)\) is a nonnegative supermartingale
starting at one.  This same pair remains feasible under \(H_0\), so
\(0\le M_t(\lambda)\le M_t(\theta_0,z_0;\lambda)\) for every \(n_c\le t\le T\).
This domination holds at stopping times and proves the e-process claim, even
when \(M_t(\lambda)\) itself is not a supermartingale.

\begin{example}
Let \(n_c=0\), \(T=2\), and \(\Theta=\{\theta_0\}\). Let \(Z_1,Z_2\)
be iid, taking values \(1,7,8\) with probabilities \(1/5,3/5,1/5\),
respectively. Use constant reports \(\tilde r=37/5\) with \(p=1/2\),
\(\varepsilon=1/100\), and \(\lambda_t=2/5\). The residual VaR is
\(z_0=7\) and ES is \(37/5\), so \(H_0\) holds. On \(\{Z_1=8\}\),
direct evaluation of the infimum gives
\(\E[M_2(\lambda)\mid Z_1=8]=\frac{301}{200}>\frac75=M_1(\lambda)\).
Thus \(M_t(\lambda)\) is not a supermartingale, although it is an
e-process by Proposition~\ref{thm:full-parameter-profile-validity}.
\end{example}

\begin{remark}
\label{rem:general-quantile-family}
The proof of Proposition~\ref{thm:full-parameter-profile-validity}
extends to a fixed family of finite predictable candidate conditional VaR
paths \(q_t(\eta)\), with the ES reports and betting process as above.
Take the infimum of the cumulative ES betting products over feasible
candidates satisfying \(q_s(\eta)<r_s\) for every \(n_c<s\le t\),
allowing further feasibility restrictions. Under the ES null, the same
argument gives an e-process whenever one fixed candidate represents the
true conditional VaR throughout monitoring and remains feasible, and the
infimum is \(\mathcal F_t\)-measurable.
\end{remark}

\subsection{The cost of an unreported VaR}

Fix the observed path and the betting process \((\lambda_t)\).  To separate the costs of taking the infimum over the residual VaR and over the
parameter \(\theta\) used to compute \(\mu_t(\theta)\) and
\(\sigma_t(\theta)\), define
\[
M_t(\theta;\lambda)
=\inf_{z\le\bar z_t(\theta)}M_t(\theta,z;\lambda),
\qquad n_c<t\le T,
\]
with \(M_{n_c}(\theta;\lambda)=1\).  Thus
\(M_t(\lambda)=\inf_{\theta\in\Theta}M_t(\theta;\lambda)\).
Under the feasibility and betting conditions below, the infimum definitions
give \(0<M_t(\lambda)\le M_t(\theta_0;\lambda)\le
M_t(\theta_0,z_0;\lambda)\), with threshold-crossing events nested in the
same order, even when the reports understate the conditional ES.

\begin{proposition}
\label{prop:profiling-cost-decomposition}
Let \(\theta_0\in\Theta\), let \(\lambda\) be predictable with values in
\([0,1-p]\), and suppose that
\(z_0\le\tilde r_t(\theta_0)-\varepsilon\) for every \(n_c<t\le T\).
Then, for \(n_c\le t\le T\), the cost of profiling
the residual VaR arises only from the exceedances of the true residual VaR:
\[
\log\frac{M_t(\theta_0,z_0;\lambda)}{M_t(\theta_0;\lambda)}
\le
\sum_{s=n_c+1}^{t}
\1\{Z_s(\theta_0)>z_0\}
\log\left(
1+\frac{\lambda_sX_s(\theta_0,z_0)}{1-\lambda_s}
\right).
\]
\end{proposition}

\ProofInAppendix{proof:prop-profiling-cost-decomposition}

The VaR-profiling cost is bounded using only the dates at which the loss exceeds the true VaR,
although all reports, including those at dates with \(\lambda_s=0\), restrict
\(\bar z_t(\theta)\).  Changing \(\theta\) can change the standardized residual
\(Z_s(\theta)\) and the standardized ES forecast \(\tilde r_s(\theta)\) at every
date. The same argument therefore does not give an upper bound on the cost
of taking the infimum over \(\theta\) that uses only dates with
\(Z_s(\theta_0)>z_0\).

For the e-power comparisons that follow, fix a distribution \(F\), write
\(z_0=\VaR_p(F)\), and consider a constant report \(\tilde r\) with
\(z_0\le\tilde r-\varepsilon\) and \(\tilde r<\ES_p(F)<\infty\).
For a feasible value \(z\le\tilde r-\varepsilon\) and a constant betting process
\((\lambda_t)\) with \(\lambda_t=u\in[0,1-p]\), we follow the expected log
growth criterion of \citet[Section~5.1, equation~(12)]{wang-wang-ziegel-2026}
and define
\[
g_F(z;u)=\int\log\left\{1-u+u\,e_p^{\mathrm{ES}}(x,\tilde r,z)\right\}\,\mathrm dF(x).
\]
This is the expected logarithm of the capital multiplier for one observation,
with \(z\) and \(u\) fixed. We also define
\[
\gamma^\star
=
\inf_{z\le\tilde r-\varepsilon}\ \max_{0\le u\le1-p}g_F(z;u),
\qquad
g^\star
=
\max_{0\le u\le1-p}g_F(z_0;u).
\]
For each feasible VaR, we maximize the growth rate over constant betting processes.  Taking the infimum of these maxima gives
\(\gamma^\star\), while the maximum at the true VaR is
\(g^\star\), so \(\gamma^\star\le g^\star\).

Suppose that \(\theta_0\in\Theta\),
\(\tilde r_t(\theta_0)=\tilde r\) throughout monitoring, and
\begin{equation}
\label{eq:fixed-innovation-growth}
Z_t(\theta_0)\mid\mathcal F_{t-1}\sim F,
\qquad n_c<t\le T,
\end{equation}
where \(F\) does not depend on the date or the observed history.
Then \(\ES_p(Z_t(\theta_0)\mid\mathcal F_{t-1})=\ES_p(F)\). This assumption gives iid innovations during
monitoring; it is not needed for
Proposition~\ref{thm:full-parameter-profile-validity} or
Proposition~\ref{prop:profiling-cost-decomposition}.

This includes iid losses: taking \(\theta_0=(0,0,0,1)\) gives
\(\mu_t(\theta_0)=0\), \(\sigma_t(\theta_0)=1\), and
\(Z_t(\theta_0)=L_t\), with \(\tilde r=r\).
If we restrict the parameter set to \(\Theta=\{\theta_0\}\), then
\(M_t(\theta_0,z;\lambda)=M_t(z;\lambda)\) from
Section~\ref{sec:confidence-profile}, and
\(M_t(\lambda)=\inf_{z\le r-\varepsilon}M_t(z;\lambda)\).

\begin{proposition}
\label{prop:profiling-epower}
Under the conditions above, for every \(T>n_c\) and predictable
\(\lambda\) with values in \([0,1-p]\),
\[
\frac1{T-n_c}\E\log M_T(\lambda)
\le\gamma^\star,
\qquad
\frac1{T-n_c}\E\log M_T(\theta_0,z_0;\lambda)
\le g^\star.
\]
Moreover, \(g^\star>0\), and the bound \(g^\star\) is attained by
\(\lambda_t=u^\star\) at every monitoring date, where
\(u^\star\in\operatorname*{argmax}_{0\le u\le1-p}g_F(z_0;u)\).
\end{proposition}

\ProofInAppendix{proof:prop-profiling-epower}

The proposition bounds expected log growth per monitoring observation.
Optimizing the two betting processes separately gives the following lower
bound on the loss of total e-power:
\[
\sup_{\lambda}\E[\log M_T(\theta_0,z_0;\lambda)]
-\sup_{\lambda}\E[\log M_T(\lambda)]
\ge (T-n_c)(g^\star-\gamma^\star)\ge0.
\]
Both suprema range over predictable betting processes with values in
\([0,1-p]\).
Thus \(g^\star-\gamma^\star\) is a lower bound on the loss
per observation, and the total bound includes all \(T-n_c\) monitoring
observations.

The next result concerns power at the true parameter \(\theta_0\), with
VaR still unreported. Retain the conditions above and suppose that the
constant-report condition and \eqref{eq:fixed-innovation-growth} hold for
all \(t>n_c\).

\begin{proposition}
\label{prop:profile-asymptotic-power}
The test based on \(M_t(\theta_0;\lambda)\) has asymptotic power one for
any sufficiently small positive constant betting process \(\lambda\):
for every \(c>1\),
\[
\lim_{T\to\infty}
\Pp\!\left(\max_{n_c<t\le T}M_t(\theta_0;\lambda)\ge c\right)=1.
\]
\end{proposition}

\ProofInAppendix{proof:prop-profile-asymptotic-power}

\section{Choosing the betting process}
\label{sec:betting-selection}
We now specify the common betting process \((\lambda_t)\) used
in every candidate product and hence in \(M_t(\lambda)\).
Proposition~\ref{thm:full-parameter-profile-validity} allows any predictable
\(\lambda_t\in[0,1-p]\), but its choice affects power. We choose \(\lambda_t\)
using an empirical growth rate computed from past standardized residuals, as described below.

We prove that, under the conditions below, maximizing the expected log
growth of a single ES betting factor over constant choices \(0\le u<1\)
gives an optimizer smaller than \(1-p\). The larger
range for \(u\) is used only in this comparison; our tests use
\(\lambda_t\in[0,1-p]\).

\begin{proposition}
\label{prop:constant-growth-optimum}
Let \(z_0=\VaR_p(F)\) and \(z_0<\tilde r<\ES_p(F)<\infty\). Then
the expected log-growth function \(u\mapsto g_F(z_0;u)\), defined by the
same formula for \(0\le u<1\), has a unique maximizer
\(u^\star\in(0,1-p)\).
\end{proposition}

\ProofInAppendix{proof:prop-constant-growth-optimum}

\begin{example}
\label{ex:betting-upper-bound}
The endpoint \(1-p\) is sharp: let \(F\) put mass \(p\) at \(z_0\) and mass
\(1-p\) at \(z_0+(1-p)(\tilde r-z_0)K\) with \(K>(1-p)^{-1}\).  Then
\(z_0=\VaR_p(F)\) and \(\ES_p(F)>\tilde r\), \(e_p^{\mathrm{ES}}(x,\tilde r,z_0)\)
equals \(K\) at the upper point and zero at \(z_0\), and
\(u^\star=\frac{(1-p)K-1}{K-1}\longrightarrow 1-p\) as \(K\to\infty\).
Hence no interval \([0,b]\) with \(b<1-p\) contains an optimal constant value of the
betting process \((\lambda_t)\) for every alternative in this class.
\end{example}
This gives a one-step justification for the cap.
We choose \((\lambda_t)\) using the standardized residuals \(Z_t(\theta)\)
and ES forecasts \(\tilde r_t(\theta)\) computed on a fixed finite
parameter grid \(\mathcal G\subseteq\Theta\). The numerical choices are specified in
Section~\ref{sec:numerical-implementation}.
This grid is used only to choose \((\lambda_t)\); the infimum defining
\(M_t(\lambda)\) still covers all feasible pairs \((\theta,z)\).
At each time \(t>n_c\), let \(\mathcal I_t\) be a fixed nonempty subset of
\(\{1,\ldots,t-1\}\). When choosing \(\lambda_t\), we retain one
representative of grid points with identical standardized losses at all
earlier dates and identical standardized ES forecasts through date \(t\).
The candidate sets below contain these representatives.
For each retained grid point \(\theta\),
let the empirical VaR \(\widehat z_t(\theta)\) be the sample \(p\)-quantile of
\(\{Z_s(\theta):s\in\mathcal I_t\}\), using linear interpolation at
rank \(1+(|\mathcal I_t|-1)p\). Let \(\Theta_t^{\mathrm a}\) denote the set
of grid points \(\theta\) having at least one exceedance of
\(\widehat z_t(\theta)\) and satisfying
\(\tilde r_s(\theta)>\widehat z_t(\theta)\) for every \(s\in\mathcal I_t\) with \(Z_s(\theta)>\widehat z_t(\theta)\).
The current report further restricts the grid points to
\(\mathcal C_t=\{\theta\in\Theta_t^{\mathrm a}: \tilde r_t(\theta)>\widehat z_t(\theta)\}\).
For \(\theta\in\mathcal C_t\) and \(s\in\mathcal I_t\), define
\[
\widehat X_t(\theta,s)
=\frac{\{Z_s(\theta)-\widehat z_t(\theta)\}_+}
{(1-p)\{\tilde r_t(\theta)-\widehat z_t(\theta)\}}.
\]
The current standardized ES forecast thus reevaluates each past standardized residual.

If \(\mathcal C_t\ne\varnothing\), define
\[
\widehat g_t(u) =\frac{1}{|\mathcal C_t|\,|\mathcal I_t|}
\sum_{\theta\in\mathcal C_t}\sum_{s\in\mathcal I_t}
\log\{1-u+u\widehat X_t(\theta,s)\},
\qquad
\widehat\lambda_t=\operatorname*{argmax}_{0\le u\le1-p}\widehat g_t(u).
\]
Each \((\theta,s)\) pair has equal weight, including zero and repeated
observations. This objective is an empirical growth rate that averages the
logarithmic terms above. Each retained grid point has at least one
observation at or below its empirical VaR, contributing a \(\log(1-u)\) term, so the objective is strictly concave and its maximizer is unique.

\begin{remark}
\label{rem:minimum-growth-selection}
For nonempty \(\mathcal C_t\), an alternative is to choose \(\lambda_t\) by maximizing
\[
\widehat g_t^{\min}(u)
=\min_{\theta\in\mathcal C_t}\frac1{|\mathcal I_t|}
\sum_{s\in\mathcal I_t}
\log\{1-u+u\widehat X_t(\theta,s)\},
\qquad 0\le u\le1-p,
\]
while retaining the same infimum over cumulative products in
\eqref{eq:full-parameter-profile}. Under this criterion, a single candidate
whose empirical mean of \(\widehat X_t(\theta,s)\) is at most one forces a
zero bet, even when \(\widehat g_t\) favors a positive value. If this happens
throughout monitoring, the capital remains one. This motivates using the
average objective \(\widehat g_t\).
\end{remark}

Only grid points in \(\mathcal C_t\) enter the maximization of \(\widehat g_t\).
Define
\begin{equation}
\label{eq:betting-rule}
\lambda_t=
\begin{cases}
\widehat\lambda_t,&\mathcal C_t\ne\varnothing,\\
1-p,&\mathcal C_t=\varnothing.
\end{cases}
\end{equation}
When \(\mathcal C_t\ne\varnothing\), a zero maximizer means that a positive
bet does not improve the empirical log growth, so we retain zero.
When \(\Theta_t^{\mathrm a}\ne\varnothing\) but \(\mathcal C_t=\varnothing\),
the current ES report is at or below the empirical VaR for every retained
grid point. Since ES is at least VaR, this comparison suggests
under-reporting and motivates the largest permitted bet, \(1-p\).
When \(\Theta_t^{\mathrm a}=\varnothing\), no grid point satisfies the
historical conditions needed to evaluate the objective. We use the same
bet \(1-p\) in this case, so an undefined objective alone cannot keep
the betting process at zero. The strict inequality defining \(\mathcal C_t\) ensures a
positive denominator in \(\widehat X_t\); the margin \(\varepsilon\) is used
in the feasibility constraint for \(M_t(\lambda)\).

\section{Dense-grid implementation}
\label{sec:profile-computation}

To compute \(M_t(\lambda)\) in \eqref{eq:full-parameter-profile}, we first
consider a reference calculation that evaluates the candidate products on
a finite grid of \((\theta,z)\) and takes their minimum over the feasible
grid points. This approximates the infimum using
the same products, feasibility conditions, and betting process as in the
definition of \(M_t(\lambda)\). We refine the VaR calculation by adaptive
interval subdivision. We describe the grid construction and its
convergence under refinement. The numerical resolutions are given in
Section~\ref{sec:numerical-implementation}.

\subsection{Numerical implementation}
\label{sec:numerical-implementation}

We first choose \((\lambda_t)\) on a betting grid
\(\mathcal G\) using the rule in Section~\ref{sec:betting-selection}, then
form the reference capital approximation on a separate grid \(\mathcal N\).
In the numerical comparisons, \(\mathcal G\) remains fixed when
\(\mathcal N\) is refined, so each capital calculation uses the same
betting sequence.

The betting grid \(\mathcal G\)
consists of \(\phi\in\{-0.9,-0.8,\ldots,0.9\}\),
\(\alpha\in\{0,0.1,\ldots,0.8\}\), \(\beta\in\{0,0.1,\ldots\}\) with
\(\alpha+\beta\le0.9995\), and \(v\in\{0,1,\ldots,50\}\), giving 52,326
parameter values.
We use \(\mathcal I_t=\{21,\ldots,t-1\}\), omitting the first 20 calibration
observations when estimating the empirical VaR and growth rate.

To approximate the e-process \(M_t(\lambda)\), we divide the ranges of
\(\phi,\alpha,v\) in Section~\ref{sec:residual-quantile-null} into equal
intervals, including their endpoints. At each value of \(\alpha\), we similarly divide
\([0,\min\{\beta_{\max},1-\eta_\Theta-\alpha\}]\) for \(\beta\), so every
parameter grid point satisfies the persistence constraint.
In the numerical studies, the maximum spacing is \(0.01\) in
\(\phi,\alpha,\beta\). The subdivision counts are rounded up to multiples
of eight. For the initial variance \(v\), the grid covers \([0,50]\)
with maximum spacing one in the reference calculation.

The residual VaR has no finite lower endpoint. To parameterize its range
for the reference grid, we write
\(z=\bar z_{n_c+1}(\theta)-\frac{w}{1-w}\), \(0\le w<1\),
so that \(w\) ranges over a bounded interval while \(z\) covers every value
feasible at the first monitoring date. The endpoint \(w=1\) corresponds to
\(z\to-\infty\), where each factor tends to \(1-\lambda_s+\lambda_s/(1-p)\).
We therefore set
\[
M_t(\theta,-\infty;\lambda)
=\prod_{s=n_c+1}^t\left(1-\lambda_s+\frac{\lambda_s}{1-p}\right),
\]
which does not depend on \(\theta\).
For this reference calculation, the grid \(W\) contains 129 equally
spaced values of \(w\) on \([0,1]\).

Write \(\mathcal N\) for the resulting pairs \((\theta,z)\), including
the endpoint \(z=-\infty\). We denote the numerical approximation to
\(M_t(\lambda)\) on this grid by \(M_t(\mathcal N;\lambda)\), computed as
\begin{equation}
\label{eq:dense-grid-profile}
M_t(\mathcal N;\lambda)
=\min_{\substack{(\theta,z)\in\mathcal N\\ z\le\bar z_t(\theta)}}
M_t(\theta,z;\lambda),
\end{equation}
and \(M_{n_c}(\mathcal N;\lambda)=1\). Each candidate uses the same pair
\((\theta,z)\) in every factor. Once a candidate becomes infeasible, later
reports cannot restore it. The points with \(w=1\) keep the minimum finite.

The final capital calculations use \(v\in\{0,0.1,\ldots,10^4\}\), with
the same coefficient grid and betting process. On this parameter grid,
we profile VaR continuously by adaptive interval subdivision. Write
\(A_s(\theta,z)=1-\lambda_s+\lambda_sX_s(\theta,z)\).
Each factor is nondecreasing in \(z\) when \(L_s\ge r_s\) and
nonincreasing otherwise. Thus, on a feasible interval \([a,b]\),
\[
\log M_t(\theta,z;\lambda)\ge
B_t(\theta;[a,b])
:=\sum_{s=n_c+1}^{t}
\min\{\log A_s(\theta,a),\log A_s(\theta,b)\},
\]
including \(a=-\infty\) through the limiting factors above.
Taking the minimum of interval lower bounds over a cover of
\((-\infty,\bar z_t(\theta)]\) and all parameter grid points gives
\(\underline\ell_t^{\mathrm V}\), a lower bound for continuous VaR
profiling on that grid. The smallest log capital at evaluated feasible
points gives an upper bound \(\overline\ell_t^{\mathrm V}\).
We evaluate these bounds with directed rounding and subdivide intervals
that can lower the minimum, targeting
\(\overline\ell_t^{\mathrm V}-\underline\ell_t^{\mathrm V}\le0.01\)
at every monitoring date.
These bounds control the error in continuous VaR profiling on the chosen
parameter grid. They do not bound the error from discretizing the AR-GARCH
parameters \(\theta\).

For each \((\phi,\alpha,\beta)\) on the coefficient grid, we bound the
log capital over groups of \(v\)-grid points and all feasible \(z\).
A group needs no further refinement if its lower bound is at least the
current upper bound minus the prescribed tolerance at every monitoring
date. Otherwise, we subdivide the VaR intervals and split the group of
\(v\)-values when needed. We retain the lower bound for every group and
compute all monitoring dates in a single pass through the coefficient grid.

The tables use \(\widehat M_t(\lambda)=\exp\{\underline\ell_t^{\mathrm V}\}\)
and record the first date with \(\underline\ell_t^{\mathrm V}\ge\log c\)
as detection at threshold \(c\); an upper-bound crossing alone is not
counted. Paths whose reference-grid capital never reaches \(2\) are
classified as nonrejections without refinement. All refined paths attained
the stated tolerance.
The grid specification and the chosen \((\lambda_t)\) are saved with each
result.
Appendix~\ref{app:grid-sensitivity} compares nested capital grids while
holding the betting process fixed.

Interval subdivision over \((\phi,\alpha,\beta)\) can provide rigorous
bounds, but repeated occurrences of the same parameters in the recursive
filters make these bounds too loose for practical computation on the
parameter domain considered here.

\subsection{Convergence of the grid approximation}
\label{sec:capital-approximation}

To establish convergence of the grid approximation \(M_t(\mathcal N;\lambda)\) in
Section~\ref{sec:numerical-implementation}, we consider nested refinements that
double each subdivision count and the number of finite values of \(w\).
The losses, reports, and parameter domain are held fixed. The capital
approximation below is uniform over betting vectors, allowing them to vary
across grids. The threshold-crossing conclusion uses a common betting
process.

Fix a finite monitoring endpoint \(T>n_c\), and write \(n=T-n_c\).
Suppose \(\Theta\) is nonempty and
compact, the losses and reports are finite, and the computed location and
scale paths are continuous in \(\theta\), with positive scales.
Let \(\mathcal N_1\subseteq\mathcal N_2\subseteq\cdots\) be grids whose union
is dense in \(\Theta\times[0,1]\) in the coordinates \((\theta,w)\).
Each grid includes the candidate \(z=-\infty\), corresponding to \(w=1\).

\begin{proposition}
\label{prop:dense-grid-convergence}
For every path satisfying the preceding conditions, the grid approximation
converges to the theoretical process uniformly over monitoring dates and
betting vectors:
\begin{equation}
\label{eq:dense-grid-uniform-betting}
\sup_{\lambda\in[0,1-p]^{T-n_c}}
\max_{n_c\le t\le T}
\bigl\{M_t(\mathcal N_k;\lambda)-M_t(\lambda)\bigr\}
\downarrow0.
\end{equation}
Under any probability law \(P\) satisfying these conditions almost surely,
this convergence also holds in \(L^q(P)\) for every \(1\le q<\infty\).
For a predictable betting process \(\lambda\) taking values in
\([0,1-p]\) and common to all grids, and every \(c>1\),
\[
P\!\left(\max_{n_c<t\le T}M_t(\mathcal N_k;\lambda)\ge c\right)
\downarrow
P\!\left(\max_{n_c<t\le T}M_t(\lambda)\ge c\right).
\]
Under the null of Proposition~\ref{thm:full-parameter-profile-validity},
the limiting probability is at most \(1/c\).
\end{proposition}

\ProofInAppendix{proof:prop-dense-grid-convergence}

\section{Simulation studies}
\label{sec:sim}

\subsection{Benchmark design}

The simulation study examines rejection under exact and conservative ES
forecasts, detection under ES under-reporting, and the cost of profiling.
We compare our ES-only procedures with VaR--ES e-processes and fixed-sample ES
backtests using the same simulated loss paths. All simulations use
\(p=0.975\).
In addition to benchmarks with their own betting processes \((\lambda_t)\) or
fixed-sample tests, we evaluate \(M_t(\theta_0,z_0;\lambda)\) at the
generating parameter--VaR pair using exactly the same betting process
\((\lambda_t)\) as the e-process \(M_t(\lambda)\). This comparison measures
the cost of not knowing the true parameter vector \(\theta_0\) and residual
VaR \(z_0\), with the betting process \((\lambda_t)\) held fixed.

\begin{remark}
We report detections at thresholds \(c=2,5,10\). For the exact
e-processes, these thresholds give type-I error bounds of \(50\%\), \(20\%\), and \(10\%\) over the entire monitoring
period, under the respective nulls of
Propositions~\ref{thm:iid-construction} and~\ref{thm:full-parameter-profile-validity}.
Actual error probabilities can be much lower.
These thresholds provide graded early warnings, roughly corresponding to
minor, substantial, and strong evidence against the forecasts, respectively.
\end{remark}

The tables report the numerical approximation \(\widehat M_t(\lambda)\)
from Section~\ref{sec:numerical-implementation}. The preceding type-I error bounds apply to
the limiting process \(M_t(\lambda)\), as described in
Proposition~\ref{prop:dense-grid-convergence}.

\noindent
In the time-series comparison, the GREE (growth-rate for empirical
e-statistics), GREL (growth-rate for empirical losses), and GREM
(growth-rate for empirical mixture) columns report the VaR--ES benchmarks
of \citet{wang-wang-ziegel-2026}, under the names used there. These methods
use the reported VaR and ES forecasts to compute \(e_p^{\mathrm{ES}}\).
In our simulations, we set the reported VaR equal to the true conditional
VaR. GREE uses
the preceding VaR--ES e-statistics to select its betting process
\((\lambda_t)\). GREL chooses the current value \(\lambda_t\) by recomputing the e-statistic for each preceding
loss using the current reported VaR and ES forecasts, both of
which are fixed before the current loss is observed.  GREM is the equally
weighted mixture of these two processes.
All three use all past observations from the monitoring period.  Following
\citet[Section~5.1 and Remark~7]{wang-wang-ziegel-2026}, we use the Taylor approximation to choose each betting process \((\lambda_t)\) with
values in \([0,1/2]\). Its first value is zero.

In the time-series comparison, the ESR column is the one-sided strict
intercept ESR backtest of \citet{bayer-dimitriadis-2022}, with its asymptotic
one-sided p-value.  For the iid comparison, constant ES forecasts make this
regression singular, and we use the intercept-only form in their
equation~(2.14). The ESR columns are rejection frequencies at a fixed
sample size, not probabilities of detection during the monitoring period.

The \(M_t(\theta_0,z_0;\lambda)\) column evaluates the product at the true
parameter vector \(\theta_0\), with \(v=v_0\), and residual VaR \(z_0\)
(Proposition~\ref{prop:profiling-cost-decomposition}). Each factor therefore
uses the true conditional VaR \(\mu_s(\theta_0)+\sigma_s(\theta_0)z_0\).
The betting process \((\lambda_t)\) is chosen
using the rule in Section~\ref{sec:betting-selection}. Both \(M_t(\lambda)\) and
\(M_t(\theta_0,z_0;\lambda)\) use these same values on the same loss path.
The gap between \(M_t(\theta_0,z_0;\lambda)\) and \(M_t(\lambda)\) measures
the joint cost of taking the infimum over the unknown parameters \(\theta\)
used to compute \(\mu_t(\theta)\) and \(\sigma_t(\theta)\), and over the
residual VaR. Appendix~\ref{app:profiling-comparison} separates these two
costs by also fixing \(\theta_0\) while profiling the residual VaR.
The GREE, GREL, and GREM comparisons continue to measure the
performance of complete procedures that choose different betting processes
\((\lambda_t)\).

The case \(\Theta_t^{\mathrm a}\ne\varnothing,\ \mathcal C_t=\varnothing\)
occurs on 0.43\% of simulated monitoring days,
mostly under under-reporting.

\subsection{Comparisons under iid losses}

For each loss distribution, we generate \(N=1000\) iid paths of 1200
observations. The first \(n_c=200\) observations on each path form the
calibration period. We reuse these paths and calibration samples across
monitoring horizons \(T-n_c=250,500,1000\), all three ES forecasts, and all
methods.

The constant ES forecast is \(r=m\,\ES_p(F)\), where \(F\) is the iid
loss distribution and \(m\in\{0.9,1.0,1.1\}\), corresponding to
under-reporting, exact reporting, and over-reporting.
The iid setting is a special case of the time-series setting considered here.
We take
\(\theta_0=(\phi_0,\alpha_0,\beta_0,v_0)=(0,0,0,1)\), giving
\(\mu_t(\theta_0)=0\) and \(\sigma_t(\theta_0)=1\).
Since \(\beta_0=0\), the initial value \(v_0\) does not affect these scales. When computing
\(M_t(\lambda)\), however, we use the parameter set in
\eqref{eq:simulation-parameter-set} with \(0\le v\le10^4\).
The capital grid has spacing \(0.1\) in \(v\), with the betting process
unchanged. Thus the iid experiment
also evaluates the time-series procedure without supplying these true
parameter values.

Table~\ref{tab:confidence-profile-iid} compares our two ES-only procedures
on iid losses: \(M_t^{\mathrm{cal}}(\lambda)\) from Section~\ref{sec:confidence-profile},
designed for iid losses, and the time-series procedure \(M_t(\lambda)\) from
Section~\ref{sec:argarch-profile}.
The process \(M_t^{\mathrm{cal}}(\lambda)\) uses the iid assumption to form a VaR
confidence set from the calibration observations, with \(\delta=0.001\)
and \(\Gamma=40\), as in Section~\ref{sec:confidence-profile}.
The process \(M_t(\lambda)\) instead takes the infimum over all feasible
pairs \((\theta,z)\), and the table reports its numerical approximation
\(\widehat M_t(\lambda)\). This comparison includes their different
betting rules: \(M_t^{\mathrm{cal}}(\lambda)\) fixes \(\lambda_t\) after
calibration, while \(M_t(\lambda)\) updates it as in
Section~\ref{sec:betting-selection}.

For GREE, we supply the true VaR as the VaR forecast. The
\(M_t(\theta_0,z_0;\lambda)\) column uses the same betting process as \(M_t(\lambda)\),
and ESR is the fixed-sample comparison.

\begin{table}[!htbp]
\centering
\caption{Detection percentages for iid losses (\(N=1000\)).}
\label{tab:confidence-profile-iid}
\papertablestyle
\begin{tabular*}{\linewidth}{@{\extracolsep{\fill}}lcccccc@{}}
\toprule
Distribution & \(m\) & GREE & \tablemethod{M_t^{\mathrm{cal}}}{(\lambda)} & \tablemethod{\widehat M_t}{(\lambda)} & \tablemethod{M_t}{(\theta_0,z_0;\lambda)} & ESR \\
\midrule
\multicolumn{7}{l}{\textit{250 monitoring days}} \\
\addlinespace[2pt]
Normal & $0.9$ & 46/21/10 & 63/31/15 & 49/22/11 & 71/45/30 & 86/43/20 \\
 & $1.0$ & 8/1/0 & 13/2/0 & 7/1/0 & 13/2/0 & 45/6/1 \\
 & $1.1$ & 1/0/0 & 1/0/0 & 0/0/0 & 1/0/0 & 9/0/0 \\
\addlinespace[3pt]
$t_3$ & $0.9$ & 14/4/2 & 19/4/2 & 16/5/2 & 23/8/4 & 55/14/2 \\
 & $1.0$ & 8/2/1 & 6/2/0 & 6/2/1 & 10/3/1 & 38/6/1 \\
 & $1.1$ & 4/1/0 & 3/0/0 & 3/0/0 & 5/1/0 & 24/2/0 \\
\midrule
\multicolumn{7}{l}{\textit{500 monitoring days}} \\
\addlinespace[2pt]
Normal & $0.9$ & 72/48/32 & 81/59/43 & 72/46/32 & 86/70/57 & 95/74/58 \\
 & $1.0$ & 11/2/0 & 19/3/1 & 13/2/0 & 21/4/1 & 47/13/4 \\
 & $1.1$ & 1/0/0 & 1/0/0 & 0/0/0 & 1/0/0 & 3/0/0 \\
\addlinespace[3pt]
$t_3$ & $0.9$ & 20/7/3 & 29/10/4 & 26/9/5 & 34/15/8 & 65/29/9 \\
 & $1.0$ & 10/3/1 & 10/2/1 & 9/3/1 & 14/5/2 & 40/10/2 \\
 & $1.1$ & 4/2/0 & 4/1/0 & 4/1/0 & 6/1/1 & 24/4/1 \\
\midrule
\multicolumn{7}{l}{\textit{1000 monitoring days}} \\
\addlinespace[2pt]
Normal & $0.9$ & 93/80/67 & 92/82/73 & 87/73/60 & 96/90/84 & 99/93/85 \\
 & $1.0$ & 13/2/1 & 24/5/2 & 18/3/1 & 26/6/2 & 46/15/6 \\
 & $1.1$ & 1/0/0 & 1/0/0 & 0/0/0 & 1/0/0 & 1/0/0 \\
\addlinespace[3pt]
$t_3$ & $0.9$ & 30/12/6 & 40/19/10 & 39/17/9 & 47/24/14 & 76/44/21 \\
 & $1.0$ & 12/4/2 & 14/4/1 & 12/4/2 & 19/6/3 & 44/13/4 \\
 & $1.1$ & 5/2/0 & 5/1/0 & 5/1/1 & 8/2/1 & 20/3/1 \\
\bottomrule
\end{tabular*}

\par\vspace{2pt}
\begin{minipage}{\linewidth}
\raggedright
\emph{Notes:} Each path has \(n_c=200\) calibration observations; \(T-n_c\) counts monitoring observations.
The report multipliers \(m=0.9,1.0,1.1\) indicate under-reporting, exact reporting, and over-reporting, respectively.
Sequential entries use thresholds \(2/5/10\); ESR entries give fixed-sample
rejection frequencies at levels \(1/2,1/5,1/10\).
\end{minipage}
\end{table}

As expected, Table~\ref{tab:confidence-profile-iid} shows more frequent
detection under ES under-reporting than under exact or conservative reporting.
Both ES-only procedures detect under-reporting more often as the monitoring
period lengthens, while their rejection percentages under exact and
conservative forecasts remain at most \(5\%\) at threshold \(10\).
Under-reporting is substantially easier to detect for Normal losses than
for \(t_3\) losses, a pattern shared by the benchmark methods.

At threshold \(5\), \(M_t^{\mathrm{cal}}(\lambda)\), which uses the iid
assumption, detects under-reporting more often than
\(\widehat M_t(\lambda)\) for Normal losses at all three monitoring
horizons. For \(t_3\) losses, the two procedures give similar detection
percentages, with no consistent ordering across horizons. The advantage
for Normal losses may reflect the use of calibration observations to narrow
the range of possible VaR values, although the two procedures also use
different betting rules. Evaluating
\(M_t(\theta_0,z_0;\lambda)\) at the generating parameter--VaR pair with
the same \((\lambda_t)\) also gives higher detection percentages at this
threshold, especially for Normal losses. The latter comparison illustrates
the cost of taking the minimum over the feasible parameter--VaR pairs:
evidence must accumulate against all of them before the minimum reaches
the threshold.

The difference between the two loss distributions is more pronounced than
the difference between the two ES-only procedures: with \(t_3\) losses,
evidence against the same proportional ES shortfall accumulates more
slowly. Thus ES forecasts can be monitored without VaR forecasts,
but heavier tails may require longer monitoring to reveal persistent
under-reporting.

\subsection{AR-GARCH comparisons}

We generate losses from normalized AR(1), GARCH(1,1), and
AR(1)-GARCH(1,1) models with fixed initial values and iid innovations from
the standard normal or standardized \(t_3\) distribution. Both innovation distributions \(F\)
have mean zero and variance one. At the generating parameters, the
standardized residuals therefore equal the iid innovations, so their
conditional VaR is the constant \(\VaR_p(F)\), as required by
\eqref{eq:common-conditional-var}. The generating parameter vectors
\(\theta_0=(\phi_0,\alpha_0,\beta_0,v_0)\) and the corresponding
conditional means and variances are:
\[
\begin{alignedat}{1}
\text{AR:}\quad &\theta_0=(0.30,0,0,1),\quad
\mu_t=0.30L_{t-1},\quad \sigma_t^2=1,\\[3pt]
\text{GARCH:}\quad &\theta_0=(0,0.08,0.90,1),\quad
\mu_t=0,\quad \sigma_t^2=0.02+0.08u_{t-1}^2+0.90\sigma_{t-1}^2,\\[3pt]
\text{AR-GARCH:}\quad &\theta_0=(0.10,0.08,0.90,1),\quad
\mu_t=0.10L_{t-1},\quad \sigma_t^2=0.02+0.08u_{t-1}^2+0.90\sigma_{t-1}^2.
\end{alignedat}
\]
At each date, we take the positive square root \(\sigma_t\), draw the
innovation \(Z_t\sim F\), and compute \(u_t=\sigma_tZ_t\) and
\(L_t=\mu_t+u_t\).

For each model and innovation distribution, we generate \(N=1000\) paths
starting from \(L_0=u_0=0\) and \(\sigma_0^2=v_0=1\), without burn-in.
Each path contains \(n_c=200\) calibration observations followed by
500 monitoring observations. The recursion continues through calibration
and monitoring without resetting the state. The residual calculation uses
the initialization in Section~\ref{sec:residual-quantile-null} with \(N_0=0\).

The true conditional ES is
\(\mu_t(\theta_0)+\sigma_t(\theta_0)\ES_p(F)\), where \(\theta_0\) is the
generating parameter with \(v=v_0\). We set the ES forecast to
\(r_t=\mu_t(\theta_0)+\sigma_t(\theta_0)m\ES_p(F)\), \(m\in\{0.9,1.0,1.1\}\).
The values \(m=0.9,1.0,1.1\) give under-reporting, exact reporting, and
over-reporting, respectively. They scale only the innovation-tail
contribution, leaving \(\mu_t(\theta_0)\) unchanged.

Table~\ref{tab:ts-var-es} reports detection percentages at thresholds
\(2\), \(5\), and \(10\). We select \((\lambda_t)\) as in
Section~\ref{sec:betting-selection}, and report the numerical approximation
\(\widehat M_t(\lambda)\) as in
Section~\ref{sec:profile-computation}.
The capital calculation uses \(v\in\{0,0.1,\ldots,10^4\}\), as in
Section~\ref{sec:financial-data}, with the betting process unchanged.

\begin{table}[!htbp]
\centering
\caption{Detection percentages for time-series losses (\(N=1000\), \(T-n_c=500\)).}
\label{tab:ts-var-es}
\papertablestyle
\noindent\begin{minipage}{\linewidth}
\centering
\begin{tabular*}{\linewidth}{@{\extracolsep{\fill}}lccccccc@{}}
\toprule
\multicolumn{8}{l}{\textit{Panel A: Normal innovations}} \\
\addlinespace[2pt]
Model & \(m\) & \tablemethod{\widehat M_t}{(\lambda)} & \tablemethod{M_t}{(\theta_0,z_0;\lambda)} & GREE & GREL & GREM & ESR \\
\midrule
AR & $0.9$ & 70/44/28 & 86/70/56 & 71/44/30 & 75/46/28 & 74/48/31 & 92/62/36 \\
 & $1.0$ & 11/2/1 & 21/4/1 & 9/2/1 & 19/3/1 & 12/2/1 & 40/6/1 \\
 & $1.1$ & 1/0/0 & 2/0/0 & 2/0/0 & 4/0/0 & 2/0/0 & 3/0/0 \\
\addlinespace
GARCH & $0.9$ & 70/44/28 & 87/72/56 & 73/48/29 & 72/37/18 & 75/45/26 & 93/52/29 \\
 & $1.0$ & 17/2/0 & 27/5/1 & 11/2/0 & 16/2/1 & 13/2/0 & 37/5/0 \\
 & $1.1$ & 1/0/0 & 2/0/0 & 2/0/0 & 4/0/0 & 2/0/0 & 2/0/0 \\
\addlinespace
AR-GARCH & $0.9$ & 69/42/28 & 86/68/54 & 71/46/31 & 68/35/18 & 74/44/26 & 91/54/31 \\
 & $1.0$ & 15/2/1 & 24/5/2 & 11/2/0 & 16/2/1 & 12/2/1 & 37/6/2 \\
 & $1.1$ & 1/0/0 & 2/0/0 & 2/0/0 & 4/0/0 & 2/0/0 & 3/0/0 \\
\midrule
\multicolumn{8}{l}{\textit{Panel B: Standardized \(t_3\) innovations}} \\
\addlinespace[2pt]
Model & \(m\) & \tablemethod{\widehat M_t}{(\lambda)} & \tablemethod{M_t}{(\theta_0,z_0;\lambda)} & GREE & GREL & GREM & ESR \\
\midrule
AR & $0.9$ & 27/9/4 & 37/16/8 & 18/5/2 & 24/7/2 & 22/5/2 & 64/24/8 \\
 & $1.0$ & 11/2/0 & 18/3/1 & 6/1/0 & 11/2/1 & 8/1/0 & 42/8/2 \\
 & $1.1$ & 4/0/0 & 7/1/0 & 3/0/0 & 6/1/0 & 5/0/0 & 23/2/0 \\
\addlinespace
GARCH & $0.9$ & 28/10/4 & 41/16/9 & 20/7/2 & 24/6/2 & 22/6/2 & 56/16/4 \\
 & $1.0$ & 12/2/1 & 19/5/2 & 9/2/0 & 11/2/1 & 10/2/0 & 33/5/1 \\
 & $1.1$ & 4/0/0 & 8/1/0 & 5/1/0 & 6/1/0 & 5/1/0 & 18/1/0 \\
\addlinespace
AR-GARCH & $0.9$ & 28/12/6 & 41/18/11 & 20/7/2 & 26/8/3 & 24/7/3 & 57/16/4 \\
 & $1.0$ & 14/4/2 & 20/6/3 & 9/2/0 & 14/3/1 & 12/2/0 & 34/4/1 \\
 & $1.1$ & 6/1/0 & 10/2/1 & 5/1/0 & 7/1/0 & 6/1/0 & 17/1/0 \\
\bottomrule
\end{tabular*}
\end{minipage}

\par\vspace{2pt}
\begin{minipage}{\linewidth}
\raggedright
\emph{Notes:} AR, GARCH, and AR-GARCH denote AR(1), GARCH(1,1), and AR(1)-GARCH(1,1), respectively.
The report multipliers \(m=0.9,1.0,1.1\) indicate under-reporting, exact reporting, and over-reporting, respectively.
Sequential entries use thresholds \(2/5/10\); ESR entries give fixed-sample
rejection frequencies at levels \(1/2,1/5,1/10\).
\end{minipage}
\end{table}

Table~\ref{tab:ts-var-es} shows that the main iid findings persist under
serial dependence and changing conditional volatility. As expected,
\(\widehat M_t(\lambda)\) detects under-reporting more often than exact
or conservative reporting; its rejection percentages in the latter two
cases are at most \(2\%\) at threshold \(10\).
Detection is again more frequent with Normal than with \(t_3\) innovations.
For our procedure, this distributional difference is larger than the
variation across the AR, GARCH, and AR-GARCH specifications, consistent
with the slower accumulation of evidence for \(t_3\) discussed above.

The comparison with the VaR--ES methods depends on the innovation
distribution. At threshold \(5\), our ES-only procedure and GREE have similar
detection percentages in the AR setting with Normal innovations, whereas
GREE detects more often in the GARCH and AR-GARCH settings. With \(t_3\) innovations, our procedure
detects more often than GREE in all three settings. GREL performs well for
AR losses but falls behind both our procedure and GREE in the GARCH setting
under either innovation distribution. Since GREE and GREL receive the same
VaR--ES forecasts, their different performance also shows the importance of
how past observations are used to choose \(\lambda_t\).

Using the generating parameter--VaR pair with the same \((\lambda_t)\)
raises detection at threshold \(5\) in all six settings, showing the cost
of leaving these quantities unknown. Nevertheless, ES forecasts alone can
provide useful warnings under serial dependence and changing conditional
volatility.
For under-reporting, ESR rejects more often than our procedure at the
fixed-sample levels \(1/2\) and \(1/5\), corresponding to thresholds
\(2\) and \(5\). The sequential procedures can issue a warning before the
endpoint.

Appendix~\ref{app:gamed} extends this design to ES forecasts that change from
over-reporting to under-reporting while the loss-generating process remains
unchanged, comparing \(\widehat M_t(\lambda)\) with rolling-window
benchmarks.

\section{Financial data analysis}
\label{sec:financial-data}

We use daily NASDAQ Composite closing levels \(P_t\) from January 16, 1996 through
December 31, 2021\footnote{Data are available from FRED (series NASDAQCOM).}
and define the daily loss as \(L_t=-100\log(P_t/P_{t-1})\).
Following \citet[Table~6]{wang-wang-ziegel-2026}, we obtain one-step VaR
and ES forecasts by fitting an AR(1)-GARCH(1,1) model to the most recent
500 losses at each forecast date. We consider Normal, Student \(t\),
skewed \(t\), and empirical forecasts at level \(p=0.975\).
The empirical forecasts use the distribution of losses in the preceding
500-day window. A fifth specification, denoted ES \(+10\%\), multiplies the
skewed-\(t\) ES forecast by \(1.1\), while keeping its VaR forecast unchanged.

The common monitoring period is January 3, 2005 through December 31, 2021.
Before monitoring begins, the backtests use 900 trading days of history,
divided into blocks of 400, 300, and 200 observations.
For our ES-only procedure, the first two blocks, totaling \(700=400+300\)
observations, are used only to run the recursions for \(\mu_t(\theta)\) and
\(\sigma_t(\theta)\) at each candidate \(\theta\). The final 200 observations
form the calibration period used to choose the betting process
\((\lambda_t)\). The recursions continue through calibration and monitoring.
We also use the loss immediately before these 900 observations to compute
the first \(\mu_t(\theta)\). With \(\beta\) as large as \(0.9995\), the
effect of the initial variance \(v\) can persist through the 700 preceding
observations and the 200 calibration observations. We therefore retain
\(v\in[0,V]\) when choosing \((\lambda_t)\) and taking the infimum.

The VaR--ES benchmarks use the last two blocks, totaling \(500=300+200\)
observations, as their initial betting history. These observations contain
past losses and the VaR and ES forecasts. Thus the benchmarks'
500-day history includes all 200 days of our calibration period and the
300 days immediately before it. Once monitoring begins, their betting
history rolls forward with a window of 500 observations.
The e-process \(M_t(\lambda)\) uses the losses and ES forecasts, without the
VaR forecasts. For the capital calculation, we use the simulation grid
and continuous VaR procedure, with the \(v\)-grid enlarged to
\(\{0,0.1,\ldots,10^4\}\). The betting process is unchanged.

Table~\ref{tab:nasdaq-monitoring}
reports the corresponding detection dates, in trading days from the start of
monitoring, together with the VaR--ES benchmarks.
These comparisons concern the complete procedures, which differ in their
betting processes and initial histories.
For our procedure, interpreting these detections as evidence of ES
under-reporting requires the conditional VaR structure and report margin
condition in Section~\ref{sec:residual-quantile-null}.

\begin{table}[!htbp]
\centering
\caption{Detection dates for NASDAQ Composite forecasts.}
\label{tab:nasdaq-monitoring}
\papertablestyle
\begin{tabular*}{\linewidth}{@{\extracolsep{\fill}}lcccc@{}}
\toprule
Forecast distribution & \tablemethod{\widehat M_t}{(\lambda)} & GREE & GREL & GREM \\
\midrule
Normal & 342/358/363 & 541/645/714 & 354/363/480 & 358/480/541 \\
\(t\) & 342/358/363 & 541/718/942 & 354/363/541 & 358/480/541 \\
Skewed \(t\) & 354/363/541 & 718/3174/3478 & 363/541/1345 & 480/541/2590 \\
Skewed \(t\), ES \(+10\%\) & 541/--/-- & --/--/-- & --/--/-- & --/--/-- \\
Empirical & 714/720/764 & 741/778/925 & 948/3827/-- & 757/877/934 \\
\bottomrule
\end{tabular*}

\par\vspace{2pt}
\begin{minipage}{\linewidth}
\raggedright
\emph{Notes:} Entries give trading days from the start of monitoring for thresholds \(c=2/5/10\), respectively. ``--'' indicates no detection during monitoring.
\end{minipage}
\end{table}

For Normal forecasts, \(\widehat M_t(\lambda)\) reaches threshold \(10\)
on day 363; GREE, GREL, and GREM reach it on days
714, 480, and 541, respectively.
For Student \(t\) forecasts, \(\widehat M_t(\lambda)\) reaches threshold \(10\)
on day 363; GREE, GREL, and GREM reach it on days
942, 541, and 541, respectively.

With the skewed-\(t\) ES forecast increased by \(10\%\),
\(\widehat M_t(\lambda)\) reaches threshold \(2\) on day 541,
but does not reach thresholds \(5\) or \(10\) during monitoring.
GREE, GREL, and GREM do not reach any of the three thresholds for the increased reports.

Figure~\ref{fig:nasdaq-evidence-paths} in Appendix~\ref{app:nasdaq-log-capital}
shows the corresponding log capital paths. The thin gray horizontal lines mark \(\log 2\), \(\log 5\), and
\(\log 10\).

\section{Conclusion}
\label{sec:conclusion}

We develop e-backtesting methods that allow ES forecasts to be monitored
without a VaR forecast. The proposed methods use the sizes of
observed tail losses to accumulate evidence of ES under-reporting, while
controlling false rejections at any stopping time under the stated
assumptions. Profiling accounts for
the unreported VaR, and our theoretical comparisons quantify its cost in
log wealth and maximal expected log growth.

Simulations using the proposed betting process and numerical approximation
show that the methods can detect ES under-reporting. The comparisons also
show how the loss distribution and betting process affect detection. The
NASDAQ application further illustrates their usefulness for monitoring
daily ES forecasts.

Future work could study betting processes and numerical approximations
for other conditional VaR families.

\begingroup
\setlength{\bibsep}{0pt}

\endgroup

\clearpage
\appendix
\pdfbookmark[0]{Appendix}{appendices}
\section*{Appendices}

\section{Sensitivity to grid refinement}
\label{app:grid-sensitivity}

We examine how refining the capital grid affects the numerical
approximation \(\widehat M_t(\lambda)\) in the iid, AR, GARCH, and AR-GARCH
simulations of Section~\ref{sec:sim}. We halve and double the subdivision
counts in \((\phi,\alpha,\beta)\), giving three nested grids with maximum
spacings \(0.02\), \(0.01\), and \(0.005\). All other numerical settings
are as in the main simulations. Each calculation uses the same loss paths,
ES forecasts, and betting process \((\lambda_t)\).

Table~\ref{tab:grid-sensitivity} reports lost detections: paths that reach
a detection threshold on the coarser grid but not on the finer grid by the
same monitoring horizon.
We retain the path sets selected in preliminary calculations.
The \(0.02\)-to-\(0.01\) comparison uses the simulation paths retained
for continuous VaR profiling. For the \(0.01\)-to-\(0.005\) comparison,
we first selected at most 50 paths for each combination of model,
innovation distribution, and report multiplier \(m\), giving priority to
detected paths close to a threshold. We then included additional detected
paths whose maximum log capital up to a monitoring horizon exceeded
\(\log c\) by less than either the largest decrease observed for the same
model, innovation distribution, and \(m\) in this initial comparison or
the decrease on that path in the \(0.02\)-to-\(0.01\) comparison.
These path sets are held fixed, and each grid uses the betting process in
Section~\ref{sec:betting-selection}. Each entry gives the number of lost
detections among the paths compared as a percentage of the original
\(N=1000\) paths.

Among the paths compared, detection is less frequent on the finer grids,
but the differences in Table~\ref{tab:grid-sensitivity} are small.
For each path, we also compute the largest decrease in
\(\log\widehat M_t(\lambda)\) over the monitoring period.
The median and mean of these decreases are \(0.047\) and \(0.069\) for
the \(0.02\)-to-\(0.01\) comparison, and \(0.028\) and \(0.046\) for
the \(0.01\)-to-\(0.005\) comparison.

\begin{table}[p]
\centering
\caption{Lost detections under capital-grid refinement (\%).}
\label{tab:grid-sensitivity}
\papertablestyle
\setlength{\tabcolsep}{4pt}
\renewcommand{\arraystretch}{1.0}
\begin{tabular*}{\linewidth}{@{\extracolsep{\fill}}llrrcc@{}}
\toprule
 & & & & \multicolumn{1}{c}{\(0.02\to0.01\)} & \multicolumn{1}{c}{\(0.01\to0.005\)} \\
\cmidrule(lr){5-5}\cmidrule(l){6-6}
Model & Distribution & \(H\) & \(m\) & Lost (\%) & Lost (\%) \\
\midrule
iid & Normal & 250 & 0.9 & 0.6/0.4/0.2 & 0.5/0.3/0.1 \\
 &  &  & 1.0 & 0.0/0.0/0.0 & 0.2/0.0/0.0 \\
 &  &  & 1.1 & 0.0/0.0/0.0 & 0.0/0.0/0.0 \\
\addlinespace[2pt]
 &  & 500 & 0.9 & 0.7/0.2/0.6 & 0.2/0.6/0.3 \\
 &  &  & 1.0 & 0.8/0.0/0.0 & 0.3/0.0/0.0 \\
 &  &  & 1.1 & 0.0/0.0/0.0 & 0.0/0.0/0.0 \\
\addlinespace[2pt]
 &  & 1000 & 0.9 & 0.2/0.2/0.6 & 0.1/0.2/0.1 \\
 &  &  & 1.0 & 0.3/0.2/0.0 & 0.3/0.0/0.0 \\
 &  &  & 1.1 & 0.0/0.0/0.0 & 0.0/0.0/0.0 \\
\addlinespace[2pt]
 & \(t_3\) & 250 & 0.9 & 0.2/0.2/0.0 & 0.1/0.2/0.0 \\
 &  &  & 1.0 & 0.1/0.2/0.0 & 0.0/0.0/0.0 \\
 &  &  & 1.1 & 0.0/0.0/0.0 & 0.0/0.0/0.0 \\
\addlinespace[2pt]
 &  & 500 & 0.9 & 0.6/0.3/0.0 & 0.4/0.1/0.0 \\
 &  &  & 1.0 & 0.1/0.0/0.0 & 0.0/0.0/0.0 \\
 &  &  & 1.1 & 0.0/0.1/0.0 & 0.0/0.0/0.0 \\
\addlinespace[2pt]
 &  & 1000 & 0.9 & 0.3/0.6/0.2 & 0.4/0.1/0.1 \\
 &  &  & 1.0 & 0.2/0.1/0.0 & 0.0/0.0/0.0 \\
 &  &  & 1.1 & 0.0/0.1/0.0 & 0.0/0.0/0.0 \\
\addlinespace[3pt]
AR & Normal & 500 & 0.9 & 0.4/0.3/0.3 & 0.0/0.9/0.4 \\
 &  &  & 1.0 & 0.3/0.1/0.0 & 0.0/0.1/0.1 \\
 &  &  & 1.1 & 0.1/0.0/0.0 & 0.1/0.0/0.0 \\
\addlinespace[2pt]
 & \(t_3\) & 500 & 0.9 & 0.1/0.1/0.1 & 0.2/0.1/0.1 \\
 &  &  & 1.0 & 0.2/0.2/0.1 & 0.5/0.0/0.0 \\
 &  &  & 1.1 & 0.1/0.0/0.0 & 0.0/0.0/0.0 \\
\addlinespace[3pt]
GARCH & Normal & 500 & 0.9 & 0.3/1.0/0.7 & 0.3/0.3/0.2 \\
 &  &  & 1.0 & 0.1/0.0/0.0 & 0.1/0.1/0.0 \\
 &  &  & 1.1 & 0.0/0.0/0.0 & 0.1/0.0/0.0 \\
\addlinespace[2pt]
 & \(t_3\) & 500 & 0.9 & 0.0/0.0/0.1 & 0.1/0.1/0.0 \\
 &  &  & 1.0 & 0.2/0.1/0.0 & 0.1/0.1/0.0 \\
 &  &  & 1.1 & 0.1/0.1/0.0 & 0.0/0.0/0.0 \\
\addlinespace[3pt]
AR-GARCH & Normal & 500 & 0.9 & 0.4/0.5/0.3 & 0.6/0.2/0.1 \\
 &  &  & 1.0 & 0.0/0.0/0.0 & 0.1/0.0/0.0 \\
 &  &  & 1.1 & 0.0/0.0/0.0 & 0.0/0.0/0.0 \\
\addlinespace[2pt]
 & \(t_3\) & 500 & 0.9 & 0.2/0.2/0.1 & 0.3/0.0/0.0 \\
 &  &  & 1.0 & 0.1/0.0/0.1 & 0.0/0.0/0.0 \\
 &  &  & 1.1 & 0.0/0.0/0.0 & 0.0/0.0/0.0 \\
\bottomrule
\end{tabular*}

\par\vspace{2pt}
\begin{minipage}{\linewidth}
\raggedright
\emph{Notes:} \(H=T-n_c\) is the monitoring horizon. Lost (\%) gives lost
detections at thresholds \(2/5/10\) as percentages of the original 1000
paths. Only selected paths are compared. The same iid paths are used at all
three horizons.
\end{minipage}
\end{table}

\clearpage
\section{The costs of parameter and VaR profiling}
\label{app:profiling-comparison}

Table~\ref{tab:profiling-comparison} separates the costs of profiling the
parameter vector and residual VaR. All three columns use the same loss
paths, ES forecasts, and betting process from Section~\ref{sec:sim}.
The middle column fixes the generating parameter vector \(\theta_0\),
including \(v=v_0\), and profiles only the residual VaR. We denote its
numerical approximation by \(\widehat M_t(\theta_0;\lambda)\), using the
continuous VaR calculation and detection convention in
Section~\ref{sec:numerical-implementation}.

\begin{table}[!htbp]
\centering
\caption{Detection percentages under parameter and VaR profiling
(\(N=1000\), \(m=0.9\), \(T-n_c=500\)).}
\label{tab:profiling-comparison}
\papertablestyle
\begin{tabular*}{\linewidth}{@{\extracolsep{\fill}}lccc@{}}
\toprule
Model & \shortstack{Both unknown\\\(\widehat M_t(\lambda)\)}
& \shortstack{\(\theta_0\) known\\\(\widehat M_t(\theta_0;\lambda)\)}
& \shortstack{Both known\\\(M_t(\theta_0,z_0;\lambda)\)} \\
\midrule
\multicolumn{4}{l}{\textit{Panel A: Normal innovations}} \\
\addlinespace[2pt]
iid & 71.8/46.3/31.8 & 77.3/53.1/37.3 & 86.4/70.3/56.7 \\
AR & 69.5/44.5/28.5 & 76.3/51.3/34.8 & 86.1/69.6/55.6 \\
GARCH & 69.7/43.9/28.0 & 78.3/54.3/36.8 & 87.4/71.6/55.7 \\
AR-GARCH & 69.2/41.6/27.7 & 77.1/49.7/34.9 & 86.4/68.5/53.8 \\
\midrule
\multicolumn{4}{l}{\textit{Panel B: Standardized \(t_3\) innovations}} \\
\addlinespace[2pt]
iid & 26.1/9.3/4.9 & 31.1/13.2/6.3 & 34.3/15.0/7.6 \\
AR & 27.4/9.4/3.9 & 34.7/14.0/5.4 & 37.2/16.3/7.5 \\
GARCH & 27.9/10.3/4.4 & 36.9/13.6/7.4 & 40.6/16.2/9.0 \\
AR-GARCH & 28.5/12.1/6.3 & 37.2/15.5/8.5 & 41.1/17.7/10.6 \\
\bottomrule
\end{tabular*}

\par\vspace{2pt}
\begin{minipage}{\linewidth}
\raggedright
\emph{Notes:} Entries correspond to thresholds \(2/5/10\).
The iid rows apply the time-series construction to iid losses.
\end{minipage}
\end{table}

Fixing \(\theta_0\) raises detection percentages in every setting, and
fixing the residual VaR as well raises them further. At threshold \(10\),
both gaps are smaller for \(t_3\) innovations than for Normal innovations.
For Normal innovations, the remaining VaR-profiling cost exceeds the
parameter-profiling cost at this threshold.

\clearpage
\section{Proofs}
\label{app:proofs}

\phantomsection\label{proof:thm-iid-construction}
\begin{proof}[Proof of Proposition~\ref{thm:iid-construction}]
The inequalities \(z_0=\VaR_p(F)\le\ES_p(F)\le r\) give \(z_0\le r\).
Since \(F(z_0)\ge p\), the coverage failure \(\{z_\delta>z_0\}\) satisfies
\[
\Pp_F(z_\delta>z_0)
=
\Pp_F\{L_{(j_\delta)}>z_0\}
\le
\Pp\{\operatorname{Bin}(n_c,p)<j_\delta\}
\le
\delta.
\]
The product \(M_t(z_0;\lambda)\) from Section~\ref{sec:confidence-profile} satisfies
\(M_{n_c}(z_0;\lambda)=1\).  At a stopping time \(n_c\le\tau\le T\), the numerator
\((1+\Gamma\delta)M_\tau^{\mathrm{cal}}(\lambda)\) in \eqref{eq:calibrated-process} equals \(\min\{\inf_{z_\delta\le z\le r}M_\tau(z;\lambda),\Gamma\}\) on \(\{z_\delta\le r\}\) and \(\Gamma\) otherwise.  Then
\[
(1+\Gamma\delta)M_\tau^{\mathrm{cal}}(\lambda)
\le
\min\{M_\tau(z_0;\lambda),\Gamma\}
+
\Gamma\1_{\{z_\delta>z_0\}}.
\]
Indeed, on \(\{z_\delta\le z_0\}\) we have \(z_\delta\le z_0\le r\), so the
construction is
in its first case with \(z_0\in[z_\delta,r]\); hence \(\inf_{z_\delta\le z\le r}M_\tau(z;\lambda)\le M_\tau(z_0;\lambda)\) and
\((1+\Gamma\delta)M_\tau^{\mathrm{cal}}(\lambda)\le\min\{M_\tau(z_0;\lambda),\Gamma\}\).  On \(\{z_\delta>z_0\}\), the left-hand side is at
most \(\Gamma\) in either case, whereas the right-hand side is at least
\(\Gamma\).

If \(z_0<r\), the Rockafellar--Uryasev representation and
\(\ES_p(F)\le r\) give
\(\E_F[e_p^{\mathrm{ES}}(L_s,r,z_0)\mid \mathcal F_{s-1}]\le1\).
If \(z_0=r\), then \(\ES_p(F)=r\) and
\(\E_F[(L_s-r)_+\mid\mathcal F_{s-1}]=0\); hence \(L_s\le r\) almost surely
and the e-statistic at \(z_0=r\) equals one.  Thus the same conditional inequality
holds in both cases.
Since \(\lambda_s\) is predictable, \(M_t(z_0;\lambda)\) is a nonnegative supermartingale
with \(M_{n_c}(z_0;\lambda)=1\).  Thus \(\E_F[M_\tau(z_0;\lambda)]\le1\), and
\(\E_F[\min\{M_\tau(z_0;\lambda),\Gamma\}]\le1\).

Taking expectations in the bound above yields
\[
(1+\Gamma\delta)\,\E_F[M_\tau^{\mathrm{cal}}(\lambda)]
\le
\E_F[\min\{M_\tau(z_0;\lambda),\Gamma\}]
+\Gamma\Pp_F(z_\delta>z_0)
\le
1+\Gamma\delta.
\]
Dividing by \(1+\Gamma\delta\) yields
\(\E_F[M_\tau^{\mathrm{cal}}(\lambda)]\le1\).
The second claim follows from the first by Markov's inequality at the
first time \(M_t^{\mathrm{cal}}(\lambda)\) reaches \(c\).
\end{proof}

\phantomsection\label{proof:thm-full-parameter-profile-validity}
\begin{proof}[Proof of Proposition~\ref{thm:full-parameter-profile-validity}]
We first prove measurability.  Fix \(t>n_c\).  The endpoint
\(\bar z_t(\theta)\) is measurable for each fixed \(\theta\) and continuous
in \(\theta\) for each fixed path.  These properties also hold for every standardized ES forecast
\(\tilde r_s(\theta)\).  Parametrize the feasible values by
\(z=\bar z_t(\theta)-a\), \(a\ge0\).  The map
\((\theta,a)\longmapsto M_t\bigl(\theta,\bar z_t(\theta)-a;\lambda\bigr)\)
is continuous for each path, since every denominator is at least
\((1-p)\varepsilon>0\), and is \(\mathcal F_t\)-measurable at every fixed
\((\theta,a)\).  Let \(D\) be a countable dense subset of
\(\Theta\times[0,\infty)\).  Continuity gives
\(M_t(\lambda)=\inf_{(\theta,a)\in D}M_t\bigl(\theta,\bar z_t(\theta)-a;\lambda\bigr)\),
which is a countable infimum of measurable functions.  Neither compactness
nor attainment of the infimum is needed.

We next prove validity.  Under \(H_0\), choose one fixed pair
\((\theta_0,z_0)\) satisfying \eqref{eq:common-conditional-var}
and the parameter coverage and report margin conditions.
For \(n_c<s\le T\), \(z_0\) is the conditional \(p\)-quantile of
\(Z_s(\theta_0)\).
Apply \eqref{eq:ru-lower-quantile} to this conditional distribution.  The finite
reported bound under \(H_0\) gives, almost surely,
\[
\E[(Z_s(\theta_0)-z_0)_+\mid\mathcal F_{s-1}]
=(1-p)\{\ES_p(Z_s(\theta_0)\mid\mathcal F_{s-1})-z_0\}<\infty.
\]
The denominator \(\tilde r_s(\theta_0)-z_0\) is
\(\mathcal F_{s-1}\)-measurable and at least \(\varepsilon>0\).  Hence
\[
\E[X_s(\theta_0,z_0)\mid\mathcal F_{s-1}]
=\frac{\ES_p(Z_s(\theta_0)\mid\mathcal F_{s-1})-z_0}{\tilde r_s(\theta_0)-z_0}\le1.
\]
Since \(\lambda_s\) is predictable and belongs to \([0,1-p]\), the product
\(M_t(\theta_0,z_0;\lambda)\) is a nonnegative supermartingale starting
at one.  This requires no common full conditional distribution or independence.
The pair \((\theta_0,z_0)\) remains feasible at every time \(t\), so
\(0\le M_t(\lambda)\le M_t(\theta_0,z_0;\lambda)\), \(n_c\le t\le T\).
Optional stopping yields
\(\E[M_\tau(\lambda)]\le
\E[M_\tau(\theta_0,z_0;\lambda)]\le1\).
Finally, the event that \(M_t(\lambda)\) reaches or exceeds the threshold
is contained in the corresponding event for this supermartingale, and
Ville's inequality gives the stated bound.
\end{proof}

\phantomsection\label{proof:prop-profiling-cost-decomposition}
\begin{proof}[Proof of Proposition~\ref{prop:profiling-cost-decomposition}]
At \(t=n_c\), the bound is an equality since all products equal one and
the sum is empty. Fix a time \(t>n_c\). Throughout, write
\(f_s(z)=1-\lambda_s+\lambda_sX_s(\theta_0,z)\) and
\(\bar z_t=\min_{n_c<s\le t}\tilde r_s(\theta_0)-\varepsilon\), so that
the feasible values of \(z\) at time \(t\) are \(z\le\bar z_t\), and the
report margin condition ensures that \(z_0\) is feasible at every \(t\).

Fix \(n_c<s\le t\), write \(Z_s=Z_s(\theta_0)\) and
\(\tilde r_s=\tilde r_s(\theta_0)\), and define
\(f_s^-:=\inf_{z\le\bar z_t}f_s(z)\).
For \(z<\tilde r_s\), the ratio
\((Z_s-z)/(\tilde r_s-z)\) has derivative
\((Z_s-\tilde r_s)/(\tilde r_s-z)^2\) in \(z\) and tends to one as
\(z\to-\infty\).  If \(Z_s\le\tilde r_s\), the factor
\(f_s(z)=1-\lambda_s+\lambda_s(Z_s-z)_+/\{(1-p)(\tilde r_s-z)\}\) is
therefore nonincreasing on \((-\infty,\bar z_t]\), and its infimum is
\(f_s(\bar z_t)=f_s^-\).  If \(Z_s>\tilde r_s\), the factor is nondecreasing
and tends to \(1-\lambda_s+\lambda_s/(1-p)=f_s^-\) as \(z\to-\infty\).  In
both cases, \(f_s(z)\ge f_s^-\ge1-\lambda_s\) for every feasible \(z\).
Hence the product at every feasible \(z\) satisfies
\(M_t(\theta_0,z;\lambda)\ge\prod_{s=n_c+1}^tf_s^-\), and taking the
infimum over \(z\) gives \(M_t(\theta_0;\lambda)\ge\prod_{s=n_c+1}^tf_s^-\).
Dividing \(M_t(\theta_0,z_0;\lambda)=\prod_{s=n_c+1}^tf_s(z_0)\) by this
product bounds the left-hand side by
\(\sum_{s=n_c+1}^t\log\{f_s(z_0)/f_s^-\}\).  If \(Z_s\le z_0\), then
\(Z_s\le z_0\le\bar z_t\) by the report margin condition, so
\(X_s(\theta_0,z_0)=X_s(\theta_0,\bar z_t)=0\) and
\(f_s(z_0)=f_s^-=1-\lambda_s\); these terms vanish. Hence
\[
\log\frac{M_t(\theta_0,z_0;\lambda)}{M_t(\theta_0;\lambda)}
\le
\sum_{s=n_c+1}^{t}
\1\{Z_s(\theta_0)>z_0\}
\log\frac{f_s(z_0)}{f_s^-}.
\]
The claimed bound now follows from \(f_s^-\ge1-\lambda_s\).
\end{proof}

\phantomsection\label{proof:prop-profiling-epower}
\begin{proof}[Proof of Proposition~\ref{prop:profiling-epower}]
For each feasible \(z\), the integrand in \(g_F(z;u)\) is bounded below by
\(\log p\) and above by \(u e_p^{\mathrm{ES}}(x,\tilde r,z)\).
Since \(\ES_p(F)<\infty\) implies \(\int x_+\,\mathrm dF(x)<\infty\),
dominated convergence shows that \(g_F(z;\cdot)\) is finite and continuous
on \([0,1-p]\), and hence attains its maximum.
The Rockafellar--Uryasev representation gives
\[
\int e_p^{\mathrm{ES}}(x,\tilde r,z_0)\,\mathrm dF(x)
=\frac{\ES_p(F)-z_0}{\tilde r-z_0}>1.
\]
Differentiation at zero, again justified by dominated convergence, yields
\[
g_F'(z_0;0)=\int e_p^{\mathrm{ES}}(x,\tilde r,z_0)\,\mathrm dF(x)-1>0.
\]
Thus \(g^\star>0\), and \(\gamma^\star\le g^\star\) because \(z_0\) is feasible.

For each \(z\le\tilde r-\varepsilon\), the constant standardized ES forecast gives
\(X_s(\theta_0,z)=e_p^{\mathrm{ES}}(Z_s(\theta_0),\tilde r,z)\), \(n_c<s\le T\).
Since \(\lambda_s\) is \(\mathcal F_{s-1}\)-measurable,
\eqref{eq:fixed-innovation-growth} gives
\[
\E\!\left[
\log\{1-\lambda_s+\lambda_s X_s(\theta_0,z)\}
\mid\mathcal F_{s-1}\right]
=g_F(z;\lambda_s)
\le\max_{0\le u\le1-p}g_F(z;u).
\]
Summing over the monitoring dates gives
\(\E[\log M_T(\theta_0,z;\lambda)] \le (T-n_c)\max_{0\le u\le1-p}g_F(z;u)\).
Since \(\theta_0\in\Theta\) and the report is constant at \(\theta_0\),
every such \(z\) remains feasible throughout monitoring, so
\(M_T(\lambda)\le M_T(\theta_0,z;\lambda)\).
Taking the infimum over \(z\) gives the bound \((T-n_c)\gamma^\star\).
At \(z=z_0\), the bound \((T-n_c)g^\star\) is attained by
\(\lambda_s=u^\star\). Dividing by \(T-n_c\) completes the proof.
\end{proof}

\phantomsection\label{proof:prop-profile-asymptotic-power}
\begin{proof}[Proof of Proposition~\ref{prop:profile-asymptotic-power}]
Write \(q=1-p\), let \(Z\sim F\), and put \(A=[0,\varepsilon^{-1}]\).
For \(a\in A\), define
\[
e_a(x)=\frac{\{1+a(x-\tilde r)\}_+}{q},
\qquad
B(x)=\frac{1+(x-\tilde r)_+/\varepsilon}{q}.
\]
For \(a>0\), setting \(z=\tilde r-a^{-1}\) gives
\(e_a(x)=e_p^{\mathrm{ES}}(x,\tilde r,z)\), while \(a=0\)
corresponds to the limit as \(z\to-\infty\).
We have \(0\le e_a(x)\le B(x)\) and \(\E B(Z)<\infty\), since
\(\ES_p(F)<\infty\). The Rockafellar--Uryasev representation gives,
for every \(a>0\),
\[
\E e_a(Z)
\ge 1+a\{\ES_p(F)-\tilde r\}>1.
\]
Also, \(\E e_0(Z)=1/q>1\). By dominated convergence,
\(a\mapsto\E e_a(Z)\) is continuous on the compact set \(A\), so
\(d:=\min_{a\in A}\{\E e_a(Z)-1\}>0\).

For \(0<u\le q\), write
\(h_u(a,x)=\log\{1-u+u e_a(x)\}\).
For each fixed \(x\), continuity and boundedness of \(e_a(x)\) on \(A\)
give \(h_u(a,x)/u\to e_a(x)-1\) uniformly in \(a\) as \(u\downarrow0\).
The bound \(\lvert h_u(a,x)\rvert/u\le\{1+B(x)\}/p\) allows
dominated convergence, yielding
\[
\sup_{a\in A}\left|
\frac{\E h_u(a,Z)}{u}-\{\E e_a(Z)-1\}\right|\longrightarrow0.
\]
Hence there is \(u_0\in(0,q]\) such that, for every \(0<u\le u_0\),
\(\min_{a\in A}\E h_u(a,Z)\ge ud/2>0\).

Fix such a \(u\) and take \(\lambda_t=u\) throughout monitoring.
The map \(a\mapsto h_u(a,x)\) is continuous on \(A\), with the
integrable bound \(\lvert h_u(a,x)\rvert\le|\log p|+uB(x)\).
Since the monitoring innovations are iid, the strong law holds uniformly:
\[
\sup_{a\in A}\left|
\frac1n\sum_{s=n_c+1}^{n_c+n}h_u(a,Z_s(\theta_0))
-\E h_u(a,Z)\right|\longrightarrow0
\quad\text{almost surely}.
\]
To see the uniformity, for each \(a\in A\) the expected value of
\(\sup_{b\in A:\,|b-a|<\delta}|h_u(b,Z)-h_u(a,Z)|\) tends to zero
as \(\delta\downarrow0\), by continuity and dominated convergence.
Cover \(A\) by finitely many such neighborhoods and apply the strong
law to their centers and these suprema. A countable sequence of
tolerances gives the displayed convergence.

The reports are constant, so the feasible VaR range remains
\(z\le\tilde r-\varepsilon\). Adding its limiting point \(a=0\)
does not change the infimum, by continuity. Thus
\[
\begin{aligned}
\frac1n\log M_{n_c+n}(\theta_0;\lambda)
&=\min_{a\in A}\frac1n\sum_{s=n_c+1}^{n_c+n}h_u(a,Z_s(\theta_0))\\
&\longrightarrow\min_{a\in A}\E h_u(a,Z)
\ge\frac{ud}{2}>0
\quad\text{almost surely}.
\end{aligned}
\]
It follows that \(M_t(\theta_0;\lambda)\to\infty\) almost surely.
For every \(c>1\), the probability of crossing \(c\) by time \(T\)
therefore tends to one as \(T\to\infty\).
\end{proof}

\begin{proof}[Proof of Proposition~\ref{prop:constant-growth-optimum}]
\phantomsection\label{proof:prop-constant-growth-optimum}
Let \(Z\sim F\) and \(X=e_p^{\mathrm{ES}}(Z,\tilde r,z_0)\), so that \(g_F(z_0;u)=\E\log(1-u+uX)\).  By \eqref{eq:ru-lower-quantile},
\(\E[X]=\frac{\ES_p(F)-z_0}{\tilde r-z_0}\in(1,\infty)\).
Hence \(\Pp(X>0)=1-F(z_0)>0\), while \(F(z_0)\ge p\) gives \(1-F(z_0)\le1-p<1\).  For \(0\le u<1\),
\(\log(1-u) \le\log(1-u+uX) \le\log(1+X) \le X\).
Since \(\E[X]<\infty\), \(g_F(z_0;\cdot)\) is finite and
continuous on \([0,1)\).  Since \(\Pp(X=0)=F(z_0)>0\), the
\(\log(1-u)\) term on this event makes \(g_F(z_0;\cdot)\) strictly concave.

For \(0<u<1\),
\(-\frac1{1-u} \le\frac{X-1}{1-u+uX} \le\frac1u\).
These deterministic bounds justify differentiation under the expectation
on compact subintervals of \((0,1)\), giving the continuous derivative in \(u\)
\[
g_F'(z_0;u)
=\E\left[\frac{X-1}{1-u+uX}\right]
=\frac1u\left\{1-\E\frac1{1-u+uX}\right\}.
\]
As \(u\downarrow0\), monotone convergence on \(\{X\ge1\}\) and
dominated convergence on \(\{X<1\}\) give
\(\lim_{u\downarrow0}g_F'(z_0;u)=\E[X]-1>0\).
At \(u=1-F(z_0)\), the event \(\{X=0\}\) contributes one to the expectation of the
reciprocal, so
\[
g_F'(z_0;1-F(z_0))
=-\frac{1}{1-F(z_0)}\E\left[
\frac{\1_{\{X>0\}}}{F(z_0)+\{1-F(z_0)\}X}\right]<0.
\]
Strict concavity makes \(g_F'(z_0;\cdot)\) strictly decreasing.  It therefore has a
unique zero \(u^\star\in(0,1-F(z_0))\), which is the unique global maximizer
on \([0,1)\).  Since \(1-F(z_0)\le1-p\), we obtain \(0<u^\star<1-p\).
\end{proof}

\begin{proof}[Proof of Proposition~\ref{prop:dense-grid-convergence}]
\phantomsection\label{proof:prop-dense-grid-convergence}
Fix a data path, write \(\Lambda_T=[0,1-p]^n\), and define
\[
D_k:=\sup_{\lambda\in\Lambda_T}
\max_{n_c\le t\le T}
\bigl\{\log M_t(\mathcal N_k;\lambda)-\log M_t(\lambda)\bigr\}.
\]
First fix \(\lambda\in\Lambda_T\) and a date \(t\).
The report margin makes all denominators
positive on the feasible set. Each factor is continuous there and bounded
below by \(1-\lambda_s\ge p\).
Every finite candidate entering the grid minimum is feasible. The value at
\(w=1\) is the limit along feasible pairs with \(z\to-\infty\), since each e-statistic \(X_s(\theta,z)\) tends to \((1-p)^{-1}\). Hence
\(M_t(\mathcal N_k;\lambda)\ge M_t(\lambda)\). Nested grids make
\(M_t(\mathcal N_k;\lambda)\) nonincreasing in \(k\).
Since every grid includes \(z=-\infty\),
\[
\begin{aligned}
p^{t-n_c}
&\le M_t(\lambda)\le M_t(\mathcal N_k;\lambda)\\
&\le\prod_{s=n_c+1}^t
\left(1-\lambda_s+\frac{\lambda_s}{1-p}\right)
\le(1+p)^{t-n_c}.
\end{aligned}
\]

To identify the limit, choose any feasible finite pair \((\theta,z)\).
For \(h>0\), the pair \((\theta,z-h)\) has strict report margin at every
date through \(t\). Its product converges to that at \((\theta,z)\) as
\(h\downarrow0\). The pair \((\theta,z-h)\) corresponds to some \(w<1\). Continuity of the
reports and of the map from \((\theta,w)\) to \(z\), together with the
strict margin, implies that sufficiently close grid pairs remain
feasible through \(t\). Density of the grids and continuity of the product
therefore give
\(\lim_{k\to\infty}M_t(\mathcal N_k;\lambda) \le M_t(\theta,z-h;\lambda)\).
Let \(h\downarrow0\) and take the infimum over feasible finite pairs.
Combining with the reverse bound proves the pointwise convergence.
The logarithms converge because the limit is positive.

To make this convergence uniform in \(\lambda\), note that feasibility gives
\[
0\le X_s(\theta,z)
\le\frac{1+(Z_s(\theta)-\tilde r_s(\theta))_+/\varepsilon}{1-p}
\le K_s<\infty,
\]
where compactness and continuity make \(K_s\) independent of the feasible
pair. This bound also covers \(z=-\infty\).
Each betting factor is at least \(p\), so the log products have a common
finite Lipschitz constant \(C_T\) on \(\Lambda_T\) in the supremum norm.
Taking infima or grid minima preserves that constant. Hence the errors
\(e_{k,t}(\lambda)=\log M_t(\mathcal N_k;\lambda)-\log M_t(\lambda)\)
are \(2C_T\)-Lipschitz, uniformly in \(k\) and \(t\).
Cover \(\Lambda_T\) by finitely many balls of radius \(h\), with centers in \(\Lambda_T\).
Pointwise convergence at their centers gives
\[
\limsup_{k\to\infty}\sup_{\lambda\in\Lambda_T}
\max_{n_c\le t\le T}e_{k,t}(\lambda)\le2C_T h.
\]
Letting \(h\downarrow0\) proves \(D_k\to0\).
Nestedness also gives \(D_k\downarrow0\).

Let \(\Delta_k\) denote the supremum in
\eqref{eq:dense-grid-uniform-betting}. The mean value theorem and the
capital bounds give
\[
0\le\Delta_k\le(1+p)^nD_k\longrightarrow0.
\]
Nestedness makes \(\Delta_k\) nonincreasing, proving
\eqref{eq:dense-grid-uniform-betting}.

For each deterministic \(\lambda\), the profile values are measurable by
the argument in the proof of
Proposition~\ref{thm:full-parameter-profile-validity} and the finite-grid
construction. Their continuity in \(\lambda\) allows the supremum defining
\(\Delta_k\) to be taken over a countable dense subset of \(\Lambda_T\), so
\(\Delta_k\) is measurable. Since \(0\le\Delta_k\le(1+p)^n\), bounded
convergence gives \(\E_P[\Delta_k^q]\to0\) for every \(1\le q<\infty\)
under any law \(P\) satisfying the path conditions almost surely.
No moment condition on losses or reports is needed.

For a common betting process \(\lambda\), the first part of the proof gives
\(Y_k=\max_{n_c<t\le T}M_t(\mathcal N_k;\lambda)
\downarrow Y=\max_{n_c<t\le T}M_t(\lambda)\).
Thus \(\{Y_k\ge c\}\downarrow\{Y\ge c\}\), proving convergence of the
threshold-crossing probabilities. Under the null of
Proposition~\ref{thm:full-parameter-profile-validity}, the limiting
probability is at most \(1/c\) for every \(c>1\) and every predictable
\(\lambda\) taking values in \([0,1-p]\).
\end{proof}

\section{Backtesting gamed ES forecasts}
\label{app:gamed}

A period of conservative reporting can delay the detection of subsequent
under-reporting: the earlier reports affect both the accumulated wealth and
the later values of the betting process \((\lambda_t)\).  \citet[Section~E]{wang-wang-ziegel-2026}
study this issue using stationary losses and changing forecasts.  Their rules for choosing the betting processes \((\lambda_t)\) use a rolling
history so that earlier over-reports gradually leave the estimation window.  Here we examine the same reporting pattern
for the e-process \(M_t(\lambda)\).

The design uses \(N=1000\) paths for each of the normal and standardized
\(t_3\) innovation distributions.  We use the AR(1)-GARCH(1,1)
specification in Section~\ref{sec:sim}, with \(\phi=0.10\),
\(\omega=0.02\), \(\alpha=0.08\), and \(\beta=0.90\).
These parameters remain fixed throughout the experiment.  Each path has
a calibration period of \(n_c=200\) observations,
and a monitoring period of \(800\) days.  The reported ES is
\[
r_t=\mu_t(\theta_0)+\sigma_t(\theta_0)m_t\ES_p(F),\qquad
m_t=\begin{cases}
1.1,&1\le t\le n_c+300,\\
0.9,&n_c+300<t\le n_c+800.
\end{cases}
\]
Thus the first under-report occurs on day 301 of the monitoring period.  As in the main simulation, the multiplier changes the term \(\sigma_t(\theta_0)\ES_p(F)\) and leaves the conditional mean \(\mu_t(\theta_0)\) unchanged.  All methods use the same
loss paths and ES forecasts.

The ES-only design uses \(M_t(\lambda)\) in \eqref{eq:full-parameter-profile}, the betting process \((\lambda_t)\) specified in
Section~\ref{sec:betting-selection}, and the same capital grid as in
Section~\ref{sec:financial-data}.
GREE, GREL, and GREM additionally observe the true conditional VaR.
Following \citet{wang-wang-ziegel-2026}, the GREE and GREL betting processes \((\lambda_t)\) use, on day \(j\) of the monitoring period, the preceding
\(\min\{500,j-1\}\) observations from that period, with each
value chosen using the Taylor approximation and capped at \(1/2\); their first values are
zero.  GREE uses the historical
VaR--ES e-statistics, GREL re-evaluates the historical losses using the
current forecast pair, and GREM averages their wealth processes with equal
initial weights.  All wealth processes start at one and continue across the reporting
change.

Table~\ref{tab:gamed-forecasts-h800} separates detections before and after
the reporting change.  Every percentage uses all 1000 paths per innovation
distribution.  A detection during the under-reporting period means that the threshold is reached
at least once on days 301--800 of the monitoring period, regardless of any detection during the over-reporting period.  Thus a path can be counted in both periods. The ES-only entries use
\(\widehat M_t(\lambda)\) computed as in the main simulation.

\begin{table}[!htbp]
\centering
\caption{Detection percentages before and after the reporting change (\(N=1000\)).}
\label{tab:gamed-forecasts-h800}
\papertablestyle
\begin{tabular*}{\linewidth}{@{\extracolsep{\fill}}lcccc@{}}
\toprule
Innovation & \tablemethod{\widehat M_t}{(\lambda)} & GREE & GREL & GREM \\
\midrule
\multicolumn{5}{l}{\textit{Panel A: Over-reporting period (days 1--300)}} \\
\addlinespace[2pt]
Normal & 1/0/0 & 2/0/0 & 2/0/0 & 2/0/0 \\
\(t_3\) & 5/1/0 & 4/1/0 & 5/0/0 & 5/0/0 \\
\midrule
\multicolumn{5}{l}{\textit{Panel B: Under-reporting period (days 301--800)}} \\
\addlinespace[2pt]
Normal & 60/36/23 & 58/29/14 & 60/28/12 & 61/30/14 \\
\(t_3\) & 24/8/4 & 13/2/1 & 25/6/2 & 19/4/1 \\
\bottomrule
\end{tabular*}

\par\vspace{2pt}
\begin{minipage}{\linewidth}
\raggedright
\emph{Notes:} Entries correspond to thresholds \(2/5/10\).
\end{minipage}
\end{table}

During days 301--800, \(\widehat M_t(\lambda)\) reaches thresholds \(2/5/10\)
on 60/36/23\% of paths with Normal innovations
and 24/8/4\% with \(t_3\) innovations.

\begin{figure}[!htbp]
\centering
\includegraphics[width=0.98\textwidth]{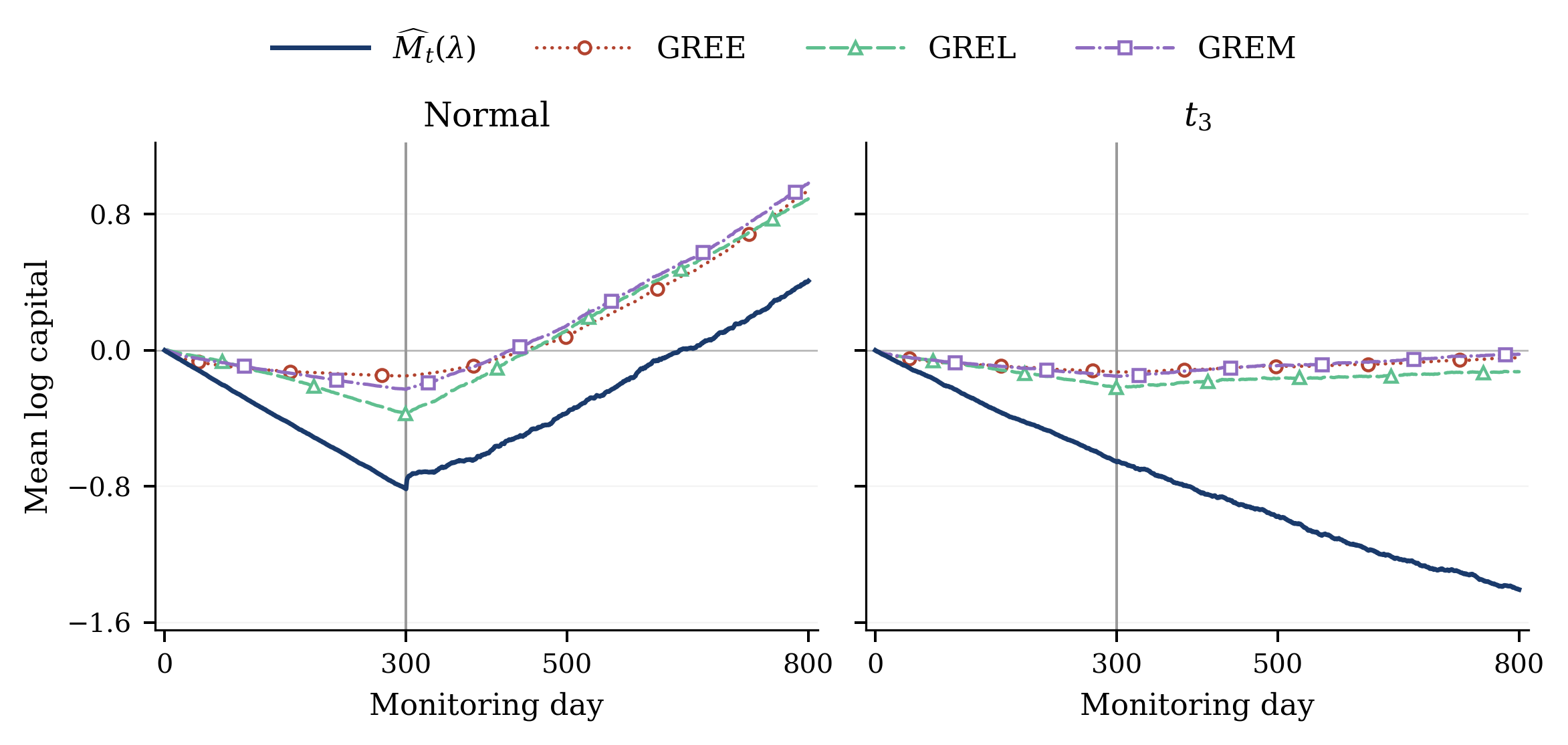}
\caption{Mean log capital before and after the reporting change (\(N=1000\)). The vertical line marks the change after day 300.}
\label{fig:gamed-log-capital}
\end{figure}

\clearpage
\section{NASDAQ log capital paths}
\label{app:nasdaq-log-capital}

Figure~\ref{fig:nasdaq-evidence-paths} illustrates the log capital paths of four methods applied to the NASDAQ data.
\begin{figure}[!htbp]
\centering
\includegraphics[width=0.98\textwidth]{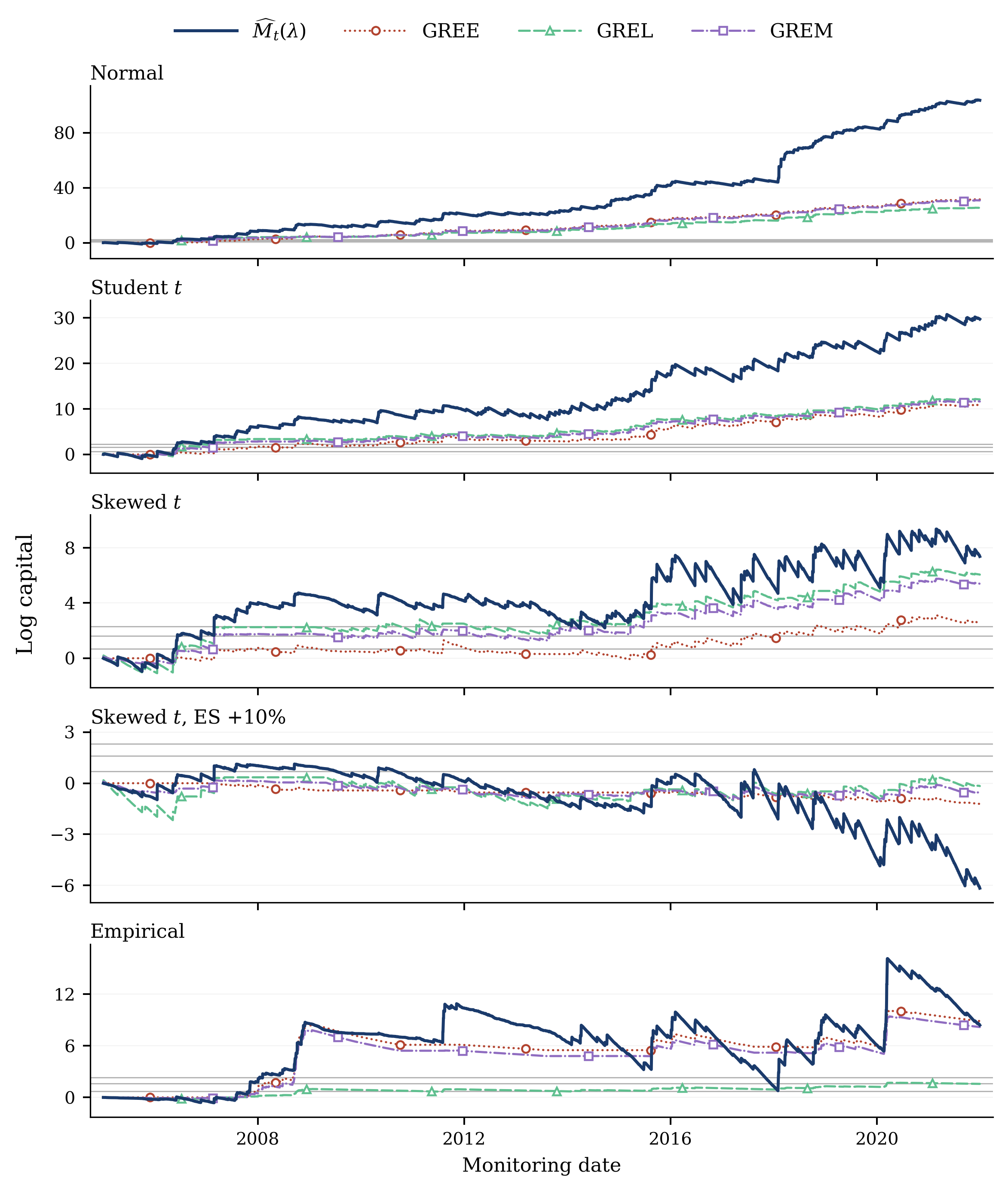}
\caption{Log capital paths for the five NASDAQ forecast specifications.}
\label{fig:nasdaq-evidence-paths}
\end{figure}

\end{document}